\documentclass{amsart}
\usepackage{amssymb}
\usepackage{amsfonts}
\usepackage{amsmath,amsthm}
\usepackage{natbib}

\usepackage{array, multirow}
\usepackage{caption}
\usepackage{subcaption}
\usepackage{setspace}
\usepackage{color}
\usepackage{xcolor}
\usepackage{breakcites}
\usepackage{amsaddr}
\usepackage{graphicx}
\usepackage[utf8]{inputenc}
\usepackage[english]{babel}
\usepackage{pifont}
\usepackage{changepage}
\usepackage{tabularx, booktabs}
\usepackage{enumitem}

\usepackage[colorlinks]{hyperref}

\usepackage{needspace}

\newtheorem{theorem}{Theorem}[section]

\newtheorem{lemma}[theorem]{Lemma}
\date{\today}
\newtheorem{corollary}[theorem]{Corollary}
\newtheorem{remark}[theorem]{Remark}
\newtheorem{assumption}[theorem]{Assumption}
\newtheorem{example}[theorem]{Example}

\newcommand{\EE}{\ensuremath{\mathbb{E}}}

\newcommand{\ul}{\ensuremath{\lfloor t/\Delta_n\rfloor}}

\newcommand{\ulT}{\ensuremath{\lfloor T/\Delta_n\rfloor}}

\newcommand{\indicator}{\ensuremath{\mathbb{I}}}

\allowdisplaybreaks

\usepackage[normalem]{ulem}

\newcommand{\edit}[1]{\textcolor{black}{#1}}

\title[Robust Functional Data Analysis for Stochastic Evolution Equations]{Robust Functional Data Analysis for Stochastic Evolution Equations in Infinite Dimensions
}



\calclayout
\author{Dennis Schroers$\dagger$}
\address{Institute of Finance and Statistics and Hausdorff Center for Mathematics\\University of Bonn,  Adenauerallee 24-26, 53113 Bonn,  Germany\\
}\thanks{$\dagger$ Email: dschroer@uni-bonn.de}
\usepackage{mathptmx}

\begin{document}
\begin{abstract}
We develop an asymptotic theory for the  jump robust measurement of covariations in the context of stochastic evolution equation in infinite dimensions.   
Namely, we identify scaling limits for realized covariations of solution processes with 
the quadratic covariation of the latent random process that drives the evolution equation which is assumed to be a general Hilbert space-valued semimartingale which could contain jumps.  We discuss applications to  dynamically consistent and outlier-robust dimension reduction in the spirit of  functional principal components and the estimation of infinite-dimensional stochastic volatility models.
\end{abstract}

\maketitle
\vspace{-1cm}

\section{Introduction}
We develop a general asymptotic theory for realized covariations in the context of  stochastic evolution equations in a Hilbert space framework which 
serve as models in various disciplines such as economics, physics, biology, or meteorology (c.f. \cite{PZ2007} and \cite{DPZ2014}).  Precisely, we consider equations 
 in a separable Hilbert space $H$ of the form
\begin{equation}\label{evol eq}
    df_t= \mathcal A f_t dt+dX_t\quad t\geq 0.
\end{equation}
Here, 
$\mathcal A$ is the generator of a strongly continuous semigroup $\mathcal S=(\mathcal S(t))_{t\geq 0}$ in $H$
 and  the latent random driver
$X$ is an $H$-valued It{\^o} semimartingale, which can be decomposed into a continuous and a jump part $X=X^C+X^J$.
In practice,  our theory can be applied to estimate the quadratic covariation of $X$, $X^C$ and $X^J$  in the context of functional data
\begin{equation}\label{functional data}
f_{i\Delta_n},\quad i=1,...,\ulT,\,\, T>0
\end{equation}
with $\Delta_n:=1/n$  for $ n\in \mathbb N$.
  Quadratic covariations are the infill-limits of the sum of tensor-squares of the increments $\Delta_i^n X= X_{i\Delta_n}-X_{(i-1)\Delta_n}$  defined by the limit (in probability and the Hilbert-Schmidt topology)
\begin{equation}\label{Quadratic variation abstract probability limit}
[X,X]_t (\cdot):=\lim_{\Delta_n \downarrow 0}\sum_{i=1}^{\ul} \langle \Delta_i^n X,  \cdot \rangle \Delta_i^n X,
\end{equation}
where $\langle \cdot, \cdot\rangle$ is the inner product of $H$.

While in general $[X,X]\neq [f,f]$,  we show that $[X,X]$, $[X^C,X^C]$, and $[X^J,X^J]$ still correspond to scaling limits of realized covariations of the observable process by employing  a particular type of increments. Namely, we consider
 the semigroup-adjusted increments 
\begin{equation}\label{semigroup adjusted increments}
    \tilde{\Delta}_i^n f:=f_{i\Delta_n}-\mathcal S(\Delta_n)f_{(i-1)\Delta_n}, \quad i=1,..., \ulT
\end{equation}
which can be recovered from the functional data \eqref{functional data} and which were examined in \cite{BSV2022} and \cite{Benth2022} for the case $X=X^C$ without jumps.   Even with discrete data, these increments can sometimes be apprximated conveniently as in \cite{Hildebrandt2023} by exploiting the spectral decomposition of the semigroup
in the context of reaction diffusions or as in \cite{Schroers2024b} by space-time differencing in the context of term structure models.
In our case with jumps we consider variants of \edit{nontruncated and truncated} \textit{semigroup adjusted realized covariations}
\begin{align}
\edit{SARCV_t^n:=}&\edit{  \sum_{i=1}^{\ul} \tilde{\Delta}_i^n f^{\otimes 2}}\label{nontruncated estimator},\\
    SARCV_t^n(u_n,-):= &\sum_{i=1}^{\ul} \tilde{\Delta}_i^n f^{\otimes 2}\indicator_{g_{n}(\tilde{\Delta}_i^n f)\leq  u_n}\quad \text{ and }\label{downward truncated estimator}\\
     SARCV_t^n(u_n,+):= &\sum_{i=1}^{\ul} \tilde{\Delta}_i^n f^{\otimes 2}\indicator_{g_{n}(\tilde{\Delta}_i^n f)>  u_n}\label{upward truncated estimator}
\end{align}  for $t\geq 0$, $n\in \mathbb N \cup\{\infty\}$ and 
\edit{a sequence $u_n= \Delta_n^{w}$, with $w\in (0,1/2)$ as well as} a sequence of subadditive truncation functions $g_{n}: H \to\mathbb R_+$, such that there are constants $c,C>0$ such that for all $f\in H$ we have
\begin{align}\label{abstract g conditions}
    c\|f\|\leq g_n(f)\leq  C\|f\|.
\end{align}

 We prove that $SARCV^n(u_n,-)$,  $SARCV^n(u_n,+)$ and $SARCV^n$ converge in probabiliy and uniformly on compacts to the quadratic covariations $[X^C,\edit{X^C}]$, $[X^J,X^J]$ and $[X,X]$ respectively under weak Assumptions.
For the estimation of $[X^C,X^C]$ via $SARCV_t^n(u_n,-)$ we  provide convergence rates and a central limit theorem under adequate regularity conditions. 
Importantly, we also present a double asymptotic theory as $T,n\to \infty $ for $ T^{-1}SARCV_T^n(u_n,-)$ to investigate a long-term limit of $[X^C,X^C]_T/T$  if it exists.

\edit{A major motivation for studying these truncated covariations is that their eigenelements provide consistent estimators for the directions of maximal variation of the noise, and can therefore be used for dimension reduction of the potentially infinite-dimensional driver \(X\).
The customary method for dimension reduction in functional data analysis is functional principal component analysis (FPCA), which typically relies on the stationary covariance under the additional assumption of weak stationarity of \(f\) (or its increments \(f_i - f_{i-1}\)).
} 
However,  \edit{in contrast to methods based on realized covariations of the driving noise},  covariance-based methods entail
consistency problems with the dynamical setting \eqref{evol eq}:
\begin{itemize}
\item[(i)] The dimensionality of a stationary covariance of $f$ (or its increments $f_{i}-f_{i-1}$) can appear high (resp. low),  although $X$ is low-dimensional (resp. high-dimensional). 
\item[(ii)]  \edit{Not all finite-dimensional parametrizations of $f$ and, hence, not all finite-di\-men\-si\-o\-nal  covariance structures, are compatible with the dynamical structure \eqref{evol eq}. Specifically, projecting $f$. onto a finite-di\-men\-si\-o\-nal linear subspace can easily eliminate the possibility that the approximated solution satisfies an equation of the form \eqref{evol eq} (cf. \cite{filipovic1999}). This subtle but undesirable effect of dimension reduction is critical since the dynamic formulation often encodes important scientific or economic principles.}
\end{itemize}
We therefore advocate a dimension-reduction technique on the basis of the (continuous) quadratic variation of $X$.
 \edit{Importantly, when $X^d$ is a finite-dimensional approximation of $X$, we can use the solution to the problem
$$
df_t^d = \mathcal{A} f_t^d dt + dX_t^d
$$
as an approximation of $f$. Unlike direct projection of $f$, this approach does not suffer from the problems described in (i) and (ii), and thus provides a dynamically consistent alternative to covariance-based methods.}

A second application is the estimation of operator-valued stochastic volatility models.
In this paper, we show how the double asymptotic estimation theory for $ T^{-1}SARCV_T^n(u_n,-)$ can be used to estimate the characteristics of a simple class of infnite-dimensional volatility models, which includes the HEIDIH model  from \cite{Petersson2022}.  This might serve as an ansatz for estimation techniques for other stationary stochastic volatility models (c.f.  \cite{BenthRudigerSuss2018}, \cite{BenthSimonsen2018}, \cite{BenthSgarra2021}, \cite{Cox2020}, \cite{Cox2021}, \cite{Petersson2022} or \cite{cox2023}).

The main theoretical novelties of this article compared to  \cite{Benth2022} and \cite{BSV2022} are the double asymptotic scheme $T,n\to\infty$ and
 the inclusion of jumps into the model, which are rarely considered in the literature on statistics of stochastic partial differential equations and significantly complicate the corresponding limit theory.
In contrast to the jump-robust estimation theory for quadratic variations of  finite-dimensional semimartigales and especially the theory of truncated realized variations of \cite{Mancini2009} and \cite{Jacod2008} some additional and inherently infinite-dimensional  challenges arise: First,  a componentwise proof of the limits does not suffice.  Instead the asymptotic theory must hold in terms of uniform convergence results in an operator norm to be applicable to the dynamically consistent functional principal component analysis and the estimation of infinite-dimensional volatility models.  Further,  the estimators must be able to account for functional outliers which  may distort  least squares techniques based on the covariation  and which might be outlying in magnitude or shape. Therefore $g_n(\cdot)=\|\cdot\|$ is not necessarily a reasonable (though a valid) choice in practice. This motivates the abstract and general form of the truncation rule of which 
a data-driven implementation can be found in \cite{Schroers2024b} in the context of bond market data.  \edit{We give a short description of this rule in the context of a simulation study.}
A further challenge, that is specific to our infinite-dimensional setting,  is that the observable process $f$ is not necessarily a semimartingale.  
A consequence is that convergence of the estimators can be much slower than $\sqrt \Delta_n$, even in the case of  finite jump activity.

Our work deviates from the majority of works on statistics for semilinear stochastic partial differential equations such as \cite{Bibinger2020}, \cite{Chong2020}, \cite{ChongDalang2020}, \cite{Hildebrandt2021}, \cite{altmeyer2021}, \cite{cialenco2022}, 
\cite{Pasemann2021}, 
 (see also \cite{Cialenco2018} for a survey) since we consider functional data and jumps, but
also from the point of view that we allow for a general differential operator $\mathcal A$.  In particular,  the choice $\mathcal A=\partial_x$ is possible,  which is, important for financial applications (c.f.  \cite{Filipovic2000}).
To allow for such generality, $\mathcal A$ needs to be known a priori, since the semigroup adjustment in   \eqref{semigroup adjusted increments} must be recovered. However,  our theory also implies that the semigroup adjustement can be dropped when $f$ is a strong solution to \eqref{evol eq}.

The article is structured as follows:
 After introducing some technical notation, Section \ref{Sec: mild ito semimartingales} introduces theory and notation for the dynamical setting \eqref{evol eq} and for $H$-valued It{\^o} semimartingales.   
The asymptotic theory for truncated semigroup adjusted realized covariations can be found in Section \ref{Sec: AbstractLimit theorems, fda case}.   In Section \ref{Sec: Identifiability}, we present the general identifiability results for the quadratic covariations of $X$, $X^C$ and $X^J$, while Section \ref{Sec: Rates and CLT} presents rates of convergence for the estimator of the quadratic variation of $X^C$.  Estimation of a long-term covariation $\lim_{T\to \infty}[X^C,X^C]_T/T$ is discussed in Section \ref{Sec: Long-time estimation}.  A simplified estimator when $f$ is a strong solution to \eqref{evol eq} is discussed in Section \ref{Sec: Strong Solution limit theorem}. 
Applications of the limit theory to dynamically consistent dimension reduction and volatility estimation can be found in  Section \ref{Sec: Applications}.  \edit{In Section \ref{Sec: Simulation Study}, we present finite-sample evidence from a simulation study analyzing the effect of semigroup adjustments on covariation measurements}.
 The formal proofs of our limit theorems can be found in the appendix.

\section{Stochastic Evolution Equations Driven by Semimartingales}\label{Sec: mild ito semimartingales}
We introduce in this section the general class of mild solutions to an equation of the form \eqref{evol eq}.  First, we define the class of Hilbert space-valued It{\^o} semimartingales \edit{which we employ throughout this article} and their quadratic variations.

\edit{Before we proceed,  we need to introduce the notation that we use throughout this work: Hereafter,  $(\Omega,\mathcal F,(\mathcal F_t)_{t\in \mathbb R_+},\mathbb P)$ is a  filtered probability space with right-continuous filtration $(\mathcal F_t)_{t\in \mathbb R_+}$ on which all subsequently occuring processes are defined. Moreover, $H$ will always denote a separable Hilbert space on which the solution to \eqref{evol eq}  takes values.  The corresponding inner product and norm are denoted by $\langle \cdot,\cdot\rangle_H$ and $\|\cdot\|_H$ and the identity operator on $H$ by $I_H.$ 
If $G$ is another separable Hilbert space
  $L(G,H)$ denotes the space of bounded linear operators from $G$ to $H$ and $L(H):=L(H,H)$.  We write $L_{\text{HS}}(G,H)$ for the Hilbert space of Hilbert-Schmidt operators from $G$ into $H$, that is $B\in L(G,H)$ such that
$
\Vert B\Vert_{L_{\text{HS}}(U,H)}^2:=\sum_{n=1}^{\infty}\| B e_n\|_H^2<\infty,
$
for an orthonormal basis $(e_n)_{n\in\mathbb N}$ of $G$. When $G=H$, we write $L_{\text{HS}}(H):=L_{\text{HS}}(H,H)$. 
Importantly,  for $h\in H$ and $g\in G$ the operator $h\otimes g:= \langle h,\cdot\rangle_H g$ is Hilbert-Schmidt and even nuclear from $H$ to $G$. Recall that $B$ is nuclear, if  $\sum_{n=1}^{\infty}\| B e_n\|_H<\infty$ for some orthonormal basis $(e_n)_{n\in\mathbb N}$ of $G$. Moreover, we shortly write $h^{\otimes 2}=h\otimes h$.
When it is clear from the context which Hilbert space is meant, we will simply write $\langle \cdot,\cdot\rangle_H=\langle \cdot,\cdot\rangle$, 
$\|\cdot\|:=\|\cdot\|_H$, $\|\cdot\|_{L(U,H)}:=\|\cdot\|_{\text{op}}$ and $\|\cdot\|_{L_{\text{HS}}(U,H)}:=\|\cdot\|_{\text{HS}}.$
Finally, for $H$-valued processes $X^n, n\in \mathbb N$, $X$, we write $X^n\stackrel{u.c.p.}{\longrightarrow}{X}\quad \text{ as }n\to\infty$ for the convergence uniformly on compacts in probability, i.e. it is $\mathbb P[\sup_{t\in [0,T]} \|X^n(t)-X(t)\|>\epsilon]\to 0$ for all $\epsilon, T>0$. 
}

\subsection{It{\^o} Semimartingales in Hilbert spaces}\label{Sec: Semimartingales in Hilbert spaces}
First, we specify the components of the driving $H$-valued It{\^o} semimartingale, which is an $H$-valued  right-continuous process with left-limits (c{\`a}dl{\`a}g) that can be decomposed as 
$$X_t:= X_t^C+X_t^J:=(A_t+M_t^C)+X_t^J\quad t\geq 0.
    $$
Here,  \edit{ we define }$A$  \edit{ to be } a continuous process of finite variation, $M^C$ is a continuous martingale and $X^J$ is another martingale modeling the jumps of $X$. We assume that $X$ is an It{\^o} semimartingale for which the components have integral representations
\begin{equation}
    A_t:=\int_0^t \alpha_s ds,\quad M_t^C:=\int_0^t\sigma_s dW_s, \quad X^J_t:=\int_0^t\int_{H\setminus \{0\}}\gamma_s(z) (N-\nu)(dz,ds). 
\end{equation}
For the first part, $(\alpha_t)_{t\geq 0}$ is an $H$-valued and  almost surely integrable (w.r.t. $\|\cdot\|$) process that is adapted to the filtration $(\mathcal F_t)_{t\geq 0}$. For the second part, the volatility process $(\sigma_t)_{t\geq 0}$ is a predictable process  taking values in the space of Hilbert-Schmidt operators $L_{\text{HS}}(U,H)$ from a separable Hilbert space $U$ into $H$ which is almost surely square-integrable, that is, $\mathbb P[\int_0^T\|\sigma_s\|_{\text{HS}}^2ds<\infty]=1$. We leave $U$ unspecified, as it is just formally the space on which the Wiener process $W$ is defined.  A cylindrical Wiener process is a weakly defined Gaussian process with independent stationary increments and covariance $I_U$, the identity on $U$. We refer to \cite{DPZ2014}, \cite{mandrekar2015} or \cite{PZ2007} for the integration theory w.r.t. cylindrical Wiener processes.
    For the third summand
    we define a homogeneous Poisson random measure $N$ on $\mathbb R_+\times H\setminus \{0\}$ and its compensator measure $\nu$ which is of the form $\nu(dz,dt)=F(dz)\otimes dt$ for a $\sigma$-finite measure $F$ on $\mathcal B(H\setminus \{0\})$. The jump volatility process $\gamma_s(z))_{s\geq 0, z\in H\setminus \{ 0\}}$  is $H$-valued, predictable and stochastically integrable w.r.t. the compensated Poisson random measure $(N-\nu)$. For a detailed account on stochastic integration w.r.t. compensated Poisson random measures in Hilbert spaces, we refer to \cite{mandrekar2015} or \cite{PZ2007}.

   \begin{remark}[On the uniqueness of the semimartingale representation]
\edit{
Since the drift $A_t = \int_0^t \alpha_s ds$ is by definition continuous,  the decomposition of $X$ into $A$, $M^C$ and $X^J$ is unique up to a $\mathbb P \otimes dt$ nullset (this follows from Corollary 3.16 and Theorem 4.18 in \cite{JacodShirayev2003}).  This implies that $\alpha$ is identified up to a $\mathbb P \otimes dt$ nullset.   Uniqueness of the coefficients $\sigma$ and $\gamma$ can obviously just be derived if we fix $W$ and $N$. 
 Independently of $W$ and $N$
more can be said when looking at the unique operator angle brackets of the martingales $M^C $ and $X^J$ (c.f.    \cite[Theorem 8.2]{PZ2007}) which are respectively $\langle\langle M^c,M^c\rangle\rangle = \int_0^t \Sigma_s ds$ and $\langle\langle M^c,M^c\rangle\rangle = \int_0^t\int_{H\setminus \{0\}} \gamma_s(z)^{\otimes 2} \nu(dz,ds)$. That implies uniqueness of $\Sigma$ up to a $\mathbb P \otimes dt$ nullset and uniqueness of $\gamma^{\otimes 2}$ up to a  $\mathbb P\otimes \nu \otimes dt$ nullset.}
 \end{remark}
 
 \begin{remark}\label{rem: ito semimartingles fin vs inf} the term It{\^o}-semimartingale in Hilbert spaces]
\edit{ 
 Our definition of an It\^o semimartingale in a Hilbert space is based on the Jacod--Grigelionis characterization of finite-dimensional It\^o semimartingales (see \cite[Theorem 2.1.2]{JacodProtter2012}).  However,  in contrast to the finite-dimensional framework, we omit a term that accounts for non-compensatable large jumps. However, we allow for large jumps as long as they can be compensated by a continuous finite-variation process. Specifically, large jump terms of the form $ \int_0^t \int_{H \setminus \{0\}} \gamma_s(z) N(dz, ds) $
are included provided that the integrability condition
$$
\int_{H \setminus \{0\}} \|\gamma_s(z)\|_H \, F(dz) < \infty
$$
holds. Under this assumption, we can decompose
$$
 \int_0^t \int_{H \setminus \{0\}} \gamma_s(z) \, N (dz, ds) := \int_0^t \int_{H \setminus \{0\}} \gamma_s(z) \, (N - \nu)(dz, ds) + \int_0^t \int_{H \setminus \{0\}} \gamma_s(z) \, F(dz) \, ds.
$$
The first part of this expression is again a discontinuous martingale of the same form as $X^J$ and  the second part is a cotinuous drift term which is of the same form as $A$.
 }
  \end{remark}

    We can now rewrite the quadratic covariation \eqref{Quadratic variation abstract probability limit} of $X$  in terms of the volatility $\sigma$ and the jumps of the process  as 
\begin{align}\label{semimartingale quadratic variation}
    [X,X]_t= \int_0^t \Sigma_s ds+ \sum_{s\leq t} (X_s-X_{s-})^{\otimes 2},
\end{align}
where $\Sigma_s = \sigma_s\sigma_s^*$ (where $\sigma_s^*$ is the Hilbert space adjoint of $\sigma_s$) and  
$X_{t-}:= \lim_{s\uparrow t}X_s$ is the left limit of $(X_t)_{t\geq 0}$ at $t$, which is well-defined, since $X_t$ has c{\`a}dl{\`a}g paths. This characterization follows as a special case of Theorem \ref{T: Abstract identification of quadr var} below \edit{(with the choice $\mathcal A\equiv 0$)}.   Moreover, it is 
$$[X^C,X^C]_t= \int_0^t \Sigma_s ds,\quad \text{and}\quad [X^J,X^J]_t=\sum_{s\leq t} (X_s-X_{s-})^{\otimes 2},$$
which clarifies, why $ [X^J,X^J]_t$ corresponds to the covariation induced by jumps.

Hilbert space-valued It{\^o} semimartingales define a very general class of drivers. In finite dimensions, It{\^o} semimartingales are commonly employed as a nonparametric framework for infill-statistics (c.f. \cite{JacodProtter2012} or \cite{Ait-SahaliaJacod2014}).  
A simple, intuitive and important special case in the infinite-dimensional setting can be found in the next Example.

\begin{example}\label{Ex: CPP in Hilbert space}[Sum of $Q$-Wiener and compound Poisson process]
A simple example for the driving semimartingale is the sum of an $H$-valued Wiener process and an $H$-valued \edit{compound Poisson process}
$$X_t= at+W^Q_t+ J_t.$$
Here, $a\in H$ and the $Q$-Wiener process $W_t\sim N(0,tQ)$ for a covariance operator $Q$ in $H$ has independent stationary increments and forms a continuous martingale.  If $Q^{\frac 12}$ denotes the positive operators square root of $Q$, we have in distribution that $W^Q_{\cdot}=\int_0^{\cdot} Q^{\frac 12} dW_t$ for a cylindrical Wiener process $W$ and we have $M_t^C=W_t^Q$ in this case. 

The jumps correspond to a compound Poisson processes
    $$J_t:=\sum_{i=1}^{N_t} \chi_i,$$
    for an i.i.d. sequence $(\chi_i)_{i\in \mathbb N}$ of random variables in $H$ with law $F$ and finite second moment ($\mathbb E[\|\chi_i\|^2]<\infty$) and a Poisson process $N$ with intensity $\lambda>0$. 
     We cannot directly set $X_t^J$ to be equal to $ J_t$, as
     it is not in compensated form and not a martingale.   However, we can rewrite the dynamics accordingly \edit{as described in Remark \ref{rem: ito semimartingles fin vs inf}}.
    For that, define the Poisson random measure
    $N(B,[0,t]):= \#\{i\leq N_t: \chi_i\in B\} $  for $ B \in \mathcal B(H\setminus \{0\}), t\geq 0.$
    This has compensator measure $\nu= \lambda dt\otimes F(dz)$ and we can write
    $X_t^J=\sum_{i=1}^{N_t} \chi_i-\lambda t \mathbb E[\chi_1]=\int_0^t \int_{H
    \setminus \{0\}} z (N-\nu)(dz,ds)$.
  Setting $A_t=(a+\lambda \mathbb E[X_1])t$, we can then define $X^C_t=A_t+M_t^C=(a+\lambda \mathbb E[X_1])t+W_t^Q$.

 The quadratic covariation of this semimartingale is
  $$[X,X]_t=[X^C,X^C]_t+[X^J,X^J]_t= t Q + \sum_{i=1}^{N_t} \chi^{\otimes 2}_i.$$

If further $H=L^2(I)$ is the space of (equivalence classes) of square-integrable real-valued functions on an interval $I\subset \mathbb R$,  we could characterize the operators $Q$ and $\chi_i^{\otimes 2}$ by integral operators
$$Qf(x)=\int_I q(x,y)f(y)dy, \quad \chi_i^{\otimes 2} f(x)=  \int_I \chi_i(x)\edit{\chi_i(y)}f(y)dy\qquad\forall x\in I,$$
for a kernel $q\in L^2(I^2)$.  So, in particular, estimation of $Q$ boils down to the estimation of a bivariate function $q$, which might be assumed to be smooth in certain applications.
\end{example}

    \subsection{Mild Solutions to Stochastic Evolution Equations}
Let us now turn to the general class of solutions to equations of the form \eqref{evol eq}.  Recall that  $\mathcal A$ generates a strongly continuous semigroup $(\mathcal S_t)_{t\geq 0}$ in $H$.  A mild solution to \eqref{evol eq} is the stochastic convolution
\begin{align}\label{mild Ito process}
f_t 
=& \mathcal S(t) f_0+\int_0^t \mathcal S(t-s)\alpha_s ds+\int_0^t \mathcal S(t-s)\sigma_sdW_s
+\int_0^t\int_{H\setminus \{0\}}\mathcal S(t-s)\gamma_s(z) (N-\nu)(dz,ds).
\end{align}
Indeed, as long as a solution to such an equation exists, the coefficients $\alpha, \sigma$ and $\gamma$ can be dependent on $f$ themselves, allowing for many nonlinear dependence patterns in the dynamics.  A comprehensive discussion on conditions ensuring the existence of such a mild solution in this case is given in \cite{FilipovicTappeTeichmann2010}.

\section{Asymptotic Theory}\label{Sec: AbstractLimit theorems, fda case}

\edit{We now turn to the estimation theory for the covariations $[X,X]$, $[X^C,X^C]$ and $[X^J,X^J]$.}

Hereafter, $f$ is a c{\`a}dl{\`a}g adapted stochastic process that is  defined on the filtered probability space $(\Omega,\mathcal F,(\mathcal F_t)_{t\in \mathbb R_+},\mathbb P)$ with right-continuous filtration $(\mathcal F_t)_{t\in \mathbb R_+}$ taking values in the separable Hilbert space $H$ and is a mild solution of type \eqref{mild Ito process} for the general evolution equation \eqref{evol eq} with an abstract generator $\mathcal A$ of a semigroup $(\mathcal S(t))_{t\geq 0}$.
 The semimartingale and its coefficients $\alpha,$ $\sigma$, $\gamma$, and drivers $W$ and $N$ are defined as in Section \ref{Sec: mild ito semimartingales}. 
 
\subsection{Identifiability of $[X,X]$, $[X^C,X^C]$, and $[X^J,X^J]$}\label{Sec: Identifiability}
In the absence of jumps,  \edit{the nontruncated} $SARCV_t^n$ is always a consistent estimator of $[X^C,X^C]_t$.   The next theorem says, that with jumps,  $SARCV$ is a consistent estimator for the quadratic variation of the driving semimartingale $X$.  \edit{Although $f$ is not necessarily a semimartingale,} no additional assumptions are needed for this identifiability result:
\begin{theorem}\label{T: Abstract identification of quadr var}
    As  $n\to \infty$ and w.r.t. the Hilbert-Schmidt norm it is
    $$ SARCV_t^n\overset{u.c.p.}{\longrightarrow}[X,X]_t=\int_0^t \Sigma_s ds+\sum_{s\leq t}\left( X_s- X_{s-}\right)^{\otimes 2}. $$
\end{theorem}
 Identification of the continuous and discontinuous part is also possible under mild conditions, reflecting those needed for finite-dimensional semimartingales such as in \cite{JacodProtter2012}. Namely, we state
 \begin{assumption}[r]\label{As: H}
    $\alpha$ is locally bounded, $\sigma$ is c{\`a}dl{\`a}g and there is a localizing sequence of stopping times $(\tau_n)_{n\in \mathbb N}$ and for each $n\in \mathbb N$ a real valued function $\Gamma_n:H\setminus \{0\}\to \mathbb R$ such that $\|\gamma_t(z)(\omega)\|\wedge 1\leq \Gamma_n(z)$ whenever $t\leq \tau_n(\omega)$ and $\int_{H\setminus \{0\}}\Gamma_n(z)^rF(dz)<\infty$ \edit{for some $r\in (0,2]$}.
\end{assumption} This assumption is a direct generalization of Assumption (H-r) in \cite{JacodProtter2012}, which is foremost an assumption on the jumps of $X$.  Indeed for $r<2$, the assumption  implies that the jumps of the process are $r$-summable, that is, we have
$$\sum_{s\leq t}\|X_s-X_{s-}\|^{l}<\infty\quad \forall l>r.$$
Assumption \ref{As: H} enables identification of $[X^C,X^C]$ and $[X^J,X^J]$:
\begin{theorem}\label{T: Identification of continuous and disc quadr var}
Let  Assumption \ref{As: H}(2) hold and let \edit{$u_n= \Delta_n^{w}$, for some $w\in (0,1/2)$}. Then  w.r.t. the Hilbert-Schmidt norm and as $n\to\infty$ it is
$$SARCV_t^n(u_n,-)\overset{u.c.p.}{\longrightarrow}[X^C,X^C]_t=\int_0^t \Sigma_s  ds$$
and 
$$SARCV_t^n(u_n,+)\overset{u.c.p.}{\longrightarrow}[X^J,X^J]_t=\sum_{s\leq t}\left( X_s- X_{s-}\right)^{\otimes 2}.$$
\end{theorem}
For identification of quadratic variations, the mild assumptions, which are in line with finite-dimensional semimartingale theory suggest that the limit theory for semigroup-adjusted quadratic variations is analogous to the one for finite-dimensional semimartingales. However,  though  general identification is possible, the difference is in the rates of convergence, and we will elaborate such rates for the jump-robust estimator $SARCV(u_n,-)$ in the next section.

\subsection{Rates of Convergence and Asymptotic Normality of  $SARCV_t^n (u_n,-)$}\label{Sec: Rates and CLT}
Rates of convergence for $SARCV(u_n,-)$ as an estimator of $[X^C,X^C]$ will depend on the regularity of the semigroup on the range of the volatility.  To make this precise, we formulate 
\begin{assumption}\label{As: spatial regularity}[$\gamma$]
One of the subsequent two conditions hold:\begin{itemize}\item[(i)]
 \begin{equation}\label{AS spatial regularity in terms of SIGMA}
        \mathbb P\left[\int_0^T \sup_{r >0}\left\|\frac{(I-\mathcal S(r))\Sigma_s}{r^{\gamma}}\right\|_{\text{HS}}ds<\infty\right]=1.
    \end{equation}
    \item[(ii)] \begin{equation}\label{AS spatial regularity in terms of SIGMAii}
     \int_0^T \sup_{r >0}   \mathbb E\left[\left\|\frac{(I-\mathcal S(r))\Sigma_s}{r^{\gamma}}\right\|_{\text{HS}}\right]ds<\infty.
    \end{equation}
\end{itemize}
\end{assumption}
\begin{remark}[Relation of Assumptions \ref{As: spatial regularity}(i) and \ref{As: spatial regularity}(ii)]
\edit{Assumption \ref{As: spatial regularity}(i) is applicable even when the volatility process $\Sigma$ does not have a finite first moment, so (ii) need not hold if (i) does. We note that we do not provide a rigorous argument to establish whether (ii) implies (i). Assumption \ref{As: spatial regularity}(ii) is included as an alternative condition because it can be easier to verify when the volatility process is stationary, a common assumption in many models.}
\end{remark}

\begin{remark} [Regularity conditions in case of fractional powers of $\mathcal A$]
\edit{Assumption \ref{As: spatial regularity} does not imply that the volatility maps into the range of $\mathcal A$ so it can in general not be expressed in terms of the generator $\mathcal A$ directly.  However when it is possible to define fractional powers of $A$ (e.g. in the case of the Laplacian) more can be said.  Let for instance $A= \sum_{i=1}^{\infty} \lambda_i e_i^{\otimes 2}$ for some eigenvalues $\lambda_i$ and eigenvectors $e_i$ for $i\in \mathbb N$. In this case,  the estimate $\|(I-S(t))A^{\eta}\|_{L(H)}\leq C t^{\eta} $ holds for a constant $C$ indepdendent of $t$ (c.f.  Lemma B.9 in \cite{Kruse2014}).  This yields that Assumption \ref{As: spatial regularity} holds, in particular, when 
 \begin{equation}
      \|A^{\gamma} \Sigma_s\|_{\text{HS}}<\infty\quad \mathbb P\otimes dt \quad\text{ almost everywhere.}
    \end{equation}}
\end{remark}
Now we can find rates of convergence:
\begin{theorem}\label{T: rates of convergence}
Let Assumptions \ref{As: H}(r) hold for some $r\in (0,2)$ and Assumption \ref{As: spatial regularity}($\gamma$) hold for some $\gamma \in (0,1/2]$ \edit{and let $u_n= \Delta_n^{w}$, for some $w\in (0,1/2)$}.
Then it is for each $\rho<(2-r)w$, $T\geq 0$ as $n\to \infty$
    $$\sup_{t\in [0,T]}\left\|SARCV_t^n(u_n,-)-[X^C,X^C]_t\right\|_{\text{HS}}=\mathcal O_p\left(\Delta_n^{\min(\gamma,\rho)}\right).$$
    In particular, if $r<2(1-\gamma)$ and if we choose $w\in [\gamma/(2-r),1/2]$ we have 
       $$\sup_{t\in [0,T]}\left\|SARCV_t^n(u_n,-)-[X^C,X^C]_t\right\|_{\text{HS}}=\mathcal O_p\left(\Delta_n^{\gamma}\right).$$
\end{theorem}
\begin{remark}
    Let us briefly comment on the rates of convergence in Theorem \ref{T: rates of convergence},  which also gives some intuition on the general discretization scheme conducted in the formal proofs.  For that, observe that we have with the notation $f'_t:= \mathcal S(t)f_0+\int_0^t \mathcal S(t-s)dX_s^C$ and $\Sigma_s^{\mathcal S_n}:= \mathcal S(\lfloor s/\Delta_n\rfloor \Delta_n-s)\Sigma_s\mathcal S(\lfloor s/\Delta_n\rfloor \Delta_n-s)^*,$ 
    that
    \begin{align*}
        & \left\|SARCV_t^n(u_n,-)-[X^C,X^C]_t\right\|_{\text{HS}}\\
        \leq &  \left\|\sum_{i=1}^{\ul} \tilde{\Delta}_i^n f^{\otimes 2}\indicator_{g_{n}(\tilde{\Delta}_i^n f)\leq  u_n}- \tilde \Delta_i^n f'^{\otimes 2}\right\|_{\text{HS}}+ \left\|\sum_{i=1}^{\ul} \left( \tilde \Delta_i^n f'^{\otimes 2}-\int_{(i-1)\Delta_n}^{i\Delta_n} \Sigma_s^{\mathcal S_n}ds\right)\right\|_{\text{HS}}\\
        &+ \left\|\int_0^T \Sigma_s^{\mathcal S_n}ds-[X^C,X^C]_t\right\|_{\text{HS}}\\
        := &
        (i)_n^t+(ii)_n^t+(iii)_n^t.
    \end{align*}
    The first summand is the estimation error that is due to the truncation of jumps and is influenced by the level of the jump activity.
     The second term is an error that behaves similarly to the approximation error of  realized variation for the quadratic variation of continuous finite-dimensional semimartingales, since it approximately corresponds to a sum of martingale differences.
      The third summand entails the error that arises due to the nonsemimartingality of $f$ and depends on the regularity of the semigroup on the range of the volatility, as assumed in Assumption \ref{As: spatial regularity}. 
Under appropriate  assumptions we then have
$$(i)_n^t=o_p(\Delta_n^{\rho}), \quad (ii)_n^t=\mathcal O_p\left(\Delta_n^{\frac 12}\right), \quad (iii)_n^t=\mathcal O_p(\Delta_n^{\gamma}).$$
\end{remark}
An infinite-dimensional central limit theorem can be derived  if we assume a regularity condition that is a bit stronger than Assumption \ref{As: spatial regularity}. Namely, we have the subsequent central limit theorem:
\begin{theorem}\label{T: CLT without disc}
Assume that 
\begin{align}\label{Abstract CLT assumption}
 \mathbb P\left[   \int_0^T\sup_{r\geq 0}\frac{\|(I-\mathcal S(r))\sigma_s\|_{\text{op}}^2}{r} ds<\infty\right]=1
\end{align}
or
\begin{align}\label{Abstract CLT assumptionii}
  \int_0^T\sup_{r\geq 0} \mathbb E\left[ \frac{\|(I-\mathcal S(r))\sigma_s\|_{\text{op}}^2}{r} \right]ds<\infty.
\end{align}
Then Assumption \ref{As: spatial regularity}($1/2$) holds. Let, moreover, Assumption \ref{As: H}(r)  hold for $r<1$,  let $w\in [1/(2-r),1/2]$. 
Then we have w.r.t. the $\|\cdot \|_{L_{\text{HS}}(H)}$ norm and as $n\to\infty$ that
  $$\sqrt n \left( SARCV_t^n(u_n,-)_t^n-[X^C,X^C]_t\right)\overset{st.}{\longrightarrow} \mathcal  N(0,\mathfrak Q_t),$$
   where $\mathcal  N(0,\mathfrak Q_t)$ is for each $t\geq 0$ a Gaussian random variable in $L_{\text{HS}}(H)$ defined on a very good filtered extension $(\tilde{\Omega},\tilde{\mathcal F},\tilde{\mathcal F}_t,\tilde {\mathbb P})$ of $(\Omega,\mathcal F,\mathcal F_t, \mathbb P)$ with mean $0$
  and covariance given for each $t\geq 0$ by a linear operator $\mathfrak Q_t:L_{\text{HS}}(H)\to  L_{\text{HS}}(H)$ such that
  $$\mathfrak Q_t  =\int_0^t \Sigma_s (\cdot+\cdot^*) \Sigma_s ds.$$
\end{theorem}

\subsection{Long-Time Estimators for Volatility}\label{Sec: Long-time estimation}

In functional data analysis, the procedure of principal component analysis is pivotal.  Since $[X^C,X^C]_t/t$ is the average of the instantaneous covariances $(\Sigma_s)_{s\in [0,t]}$, it makes sense to derive functional principal components via the eigenstructure of $[X^C,X^C]_t/t$ to recognize major modes of varation and to find low-dimensional analogs to the continuous part of the driving semimartingale $X^C$.  However, since $[X^C,X^C]_t/t$ varies with $t$,  it is difficult to find such a dimension reduction that is independent of the considered time-interval on which the data are observed.  We therefore discuss in the next section, how to estimate a stationary mean of $[X^C,X^C]_t/t$, provided that the process $(\Sigma_s)_{s\geq 0}$ is stationary and mean ergodic. That is,  we impose
\begin{assumption}\label{As: Mean stationarity and ergodicity}
We have   $\mathbb E\left[\|\sigma_s\|^2_{L_{\text{HS}}(U,L^2(\mathbb R_+))}\right]<\infty$ and  the process $(\Sigma_t)_{t\geq 0}$ is mean stationary and mean ergodic, in the sense that  there is an operator $\mathcal C$ such that for all $t$ it is $\mathcal C=\mathbb E[ \Sigma_t]$   and as $T\to\infty$ we have in probability an w.r.t. the Hilbert-Schmidt norm that 
    \begin{equation}
        \frac 1T \int_0^T\Sigma_s ds=\frac{[X^C,X^C]_T}T\to \mathcal C.
    \end{equation}
\end{assumption}
Under Assumption \ref{As: Mean stationarity and ergodicity} we have that 
$\mathbb E[ (M_t^C)^{\otimes 2}]=\mathbb E[ (\int_0^t \sigma_s dW_s)^{\otimes 2}] = t \mathcal C\quad \forall t\geq 0.  $
Hence, $\mathcal C$ is the covariance of the driving continuous martingale $M^C$ (scaled by time).
As for regular functional principal component \edit{analysis}, we can find approximately a linearly optimal finite-dimensional approximation of the driving martingale, by projecting onto the eigencomponents of $\mathcal C$. Even more, $\mathcal C$ is the instantaneous covariance of the process $f$ in the sense that
$\mathcal C=\lim_{n\to\infty}\mathbb E[(f_{t+\Delta_n}-\mathcal S(\Delta_n)f_t)^{\otimes 2}]/\Delta_n.$
To estimate $\mathcal C$, we make use of a moment assumption for the coefficients.
\begin{assumption}\label{In proof: Very very Weak localised integrability Assumption on the moments}[p,r]
For $p,r>0$ such that $\mathbb E\left[\|\gamma_s(z)\|^r\right]=\Gamma(z)$ independent of $s$ for all $s\geq 0$ and
there is a constant $A>0$ such that for all $s\geq 0$ it is
$$ \mathbb E\left[\|\alpha_s\|_{L^2(\mathbb R_+)}^p+\|\sigma_s \|_{\text{HS}}^p+\int_{H\setminus \{0\}}\|\gamma_s(z)\|^r\nu(dz)\right] \leq A.$$
\end{assumption}
Moreover, we also make an assumption on the regularity of the the volatility.
\begin{assumption}\label{In proof II}[$\gamma$]
For $\gamma\in (0,\frac 12]$ 
there is a constant $A>0$ such that for all $s\geq 0$ it is
$$ \sup_{r>0}\mathbb E\left[\frac{\|(\mathcal S(r)-I)\Sigma_s\|_{\text{HS}}}{r^{\gamma}}\right] \leq A.$$
\end{assumption}

With these assumptions,  a stationary mean $\mathcal C$ of $(\Sigma_s)_{s\geq 0}$ can be estimated.
\begin{theorem}\label{T: Long-time no disc}
 Let Assumption \ref{As: Mean stationarity and ergodicity} hold and  $\mathcal C=\mathbb E[ \Sigma_t]$ denote the global covariance of the continuous driving martingale.
Let furthermore Assumption \ref{In proof: Very very Weak localised integrability Assumption on the moments}(p,r) and \ref{In proof II}($\gamma$) hold (for the abstract semigroup $\mathcal S$) for some $r\in (0,2)$, $\gamma \in (0,1/2]$
and $p>\max(2/(1-2w),(1-wr)/(2w-rw))$.
 Then we have w.r.t. the Hilbert-Schmidt norm that as $n,T\to\infty$
$$\frac 1T SARCV_T^{n}(u_n,-)\overset{p}{\longrightarrow}\mathcal C.$$
If $r<2(1-\gamma)$ and $w\in(\gamma/(1-2w),1/2)$, $p\geq 4$ we have that 
$$  \left\|\frac 1T SARCV_T^{n}(u_n,-)-  \mathcal C\right\|_{\mathcal H}=\mathcal O_p(\Delta_n^{\gamma}+a_T).$$
\edit{where $a_T:=\mathcal O( T^{-1}\int_0^T \Sigma_s ds-\mathcal C$).}
\end{theorem}

\subsection{Limit Theorems when $f$ is a Semimartingale}\label{Sec: Strong Solution limit theorem}
Often, we might not have full knowledge of the semigroup or its generator $\mathcal A$ or it is  complicated to recover the semigroup adjustments in \eqref{semigroup adjusted increments} via discrete data.
If one is willing to assume strong regularity conditions which make $f_t$ an $H$-valued semimartingale, all previously stated results can be applied without the semigroup adjustment. The next result is a \edit{corollary} of Theorems \ref{T: Abstract identification of quadr var},  \ref{T: Identification of continuous and disc quadr var}, \ref{T: rates of convergence} and \ref{T: CLT without disc},  \ref{T: Long-time no disc}:
 \begin{corollary}\label{C: Results for realized variation}
     If $f_t\in D(\mathcal A)$ $\mathbb P\otimes dt$ almost everywhere and for all $T\geq 0$ it is $$\mathbb E\left[\int_0^T \|\sigma_s\|_{L_{\text{HS}}(U,H)}^2+\int_{H\setminus\{ 0\}} \|\gamma_s(z)\|^2 F(dz)ds\right]<\infty,$$
     then, if for $i=1,...,\ul$ it is
     $\Delta_i^n f=f_{i\Delta_n}-f_{(i-1)\Delta_n}$,  Theorems \ref{T: Abstract identification of quadr var},  \ref{T: Identification of continuous and disc quadr var}, \ref{T: rates of convergence} and \ref{T: CLT without disc},  \ref{T: Long-time no disc}, hold if we exchange $SARCV_t^n(u_n,-)$ with $$RV_t^n(u_n,-):=\sum_{i=1}^{\ul} (\Delta_i^n f)^{\otimes 2}\indicator_{g_n(Delta_i^n f)\leq  u_n}.$$
 \end{corollary}
\edit{Let us provide two examplary cases in which Corollary \ref{C: Results for realized variation} applies.}
\begin{example}
\begin{itemize}
\item[(i)] \edit{When $\mathcal A$ is continuous it is $D(\mathcal A)=H$ and the conditions of Corollary \ref{C: Results for realized variation} are trivial.}
\item[(ii)]\edit{ 
Recall the example in which $X=W^Q+J$ is the sum of a $Q$-Wiener process and a compound Poisson process from Section  \ref{Sec: Semimartingales in Hilbert spaces}. 
Moreover,  assume that $\mathcal A= \sum_{i=1}^{\infty} \mu_i e_i^{\otimes 2}$ is selfadjoint and commutes with $Q= \sum_{i=1}^{\infty} \lambda_i e_i^{\otimes 2}$ and that $f_0 \in D(\mathcal A)$ almost surely.  When $\chi_i \in D(A)$ almost surely then also $\int_0^t \int_{H\setminus \{0\}} \mathcal S(t-s) z (N-\nu)(dz,ds)=\int_0^t \int_{D(\mathcal A)\cap H\setminus (\{0\}} \mathcal S(t-s) z (N-\nu)(dz,ds)$ takes values in $D(A)$.  
Then $f$ is a strong solution if the stochastic convolution $\int_0^t \mathcal S(t-s) dW^Q_s$ takes values in $D(\mathcal A)$ $\mathbb P\otimes dt$ almost everywhere.  This is the case if $\sum_{i=1}^{\infty} \mu_i^2 \lambda_i <\infty,$ since in this case $range(Q^{\frac 12})\subset D(\mathcal A)$ and $\|AQ^{\frac 12}\|^2_{\text{HS}}<\infty$ such that Theorem 3.1(c) from \cite{GM2011} applies.  One example is the generator $$ \mathcal A = \eta_2 \partial_{xx} +\eta_1 \partial_x+ \eta_0$$ in the Hilbert space $H=L^2(0,1)$ equpped with the scalar product $\langle h,g\rangle_H:= \int_0^1 f(x)g(x) e^{x \frac {\eta_1}{\eta_2}}dx$ for which $\mu_i= \pi^2 i^2 \eta_2+\eta_1^2/(4\eta_2)-\eta_0$,  and which was considered in \cite{Bibinger2020} in the context of high-frequency parameter estimation for SPDEs.  If $Q= \mathcal A^{-\gamma}$,  the continuous stochastic convolution takes values in $D(\mathcal A)$ if $\gamma> 5/2$. }
\end{itemize}
\end{example}

Until now, we described how the quadratic variation of $X^C$ (as well as $X$ and $X^J$) can be robustly estimated with the functional data \eqref{functional data}.  We now present application\edit{s} in the next
 section.

\section{Applications}\label{Sec: Applications}
In this section we present two applications of our limit theory. The first one considers a dynamically consistent robust nonparametric method for dimension reduction and the second application showcases the estimation of stochastic volatility models in Hilbert spaces
\subsection{Dynamically Consistent Dimension Reduction}\label{Sec: Dynamics Preserving FPCA}
The infinite-dimensional nature of equations of the form \eqref{evol eq} often demand a preliminary dimension reduction.
If the dynamic relation \eqref{evol eq} encodes a scientific or economic principle, we are often interested in conserving it after such a procedure.  However, in several occasions, finite-dimensional  parametrizations for the solution process $f$ that give a supposedly good empirical fit are incompatible with the dynamics   (see e.g. \cite{filipovic1999}) and even if a parametrization is eligible in that regard, the respective dynamics of the parameters must again satisfy consistency conditions. 
This is a particularly well-known problem for term structure models in mathematical finance (cf.  \cite{bjork1999}, \cite{bjork2001}, 
\cite{filipovic2003}, 
\cite{Filipovic2000}, \cite{filipovic2000b} 
), but it is not exclusive to this application (cf. \cite{filipovic2000c}).
This is why in this section, we outline how dimension reduction can be conducted directly on the level of the random driver $X$. Precisely, we discuss the approximation of $f$ by solutions $f^d$ to equation
\begin{equation}\label{eq: low dim evol eq}
df^d_t:= \mathcal A f_t^d dt+   dX_t^d,
\end{equation}
where $X^d$ is an adequately chosen finite-dimensional approximation of $X$.
This leaves the dynamic structure imposed by the evolution equation untouched, but reduces the number of random drivers.

 For simplicity, we assume that the driving semimartingale takes the simple form of Example \ref{Ex: CPP in Hilbert space}, that is 
$$X_t= W^Q_t+ J_t,$$
 where $W^Q$ is again a $Q$-Wiener process  and $J_t=\sum_{i=1}^{N_t}\chi_i$ a compound Poisson process in $H$.  
 Since the drift $a$ in Example \ref{Ex: CPP in Hilbert space} can in general not be identified via infill asymptotics (c.f.  \cite{Koski85}),  we assume it to be $0$ for our explanation. In practice, it might be estimated when $T\to \infty$, which is a task that is left for future research.
 The jump process $J$ models rare extreme events, that can be considered as outliers and might bias least squares procedures. 
We want to approximate the continuous part $W^Q$,  for which we need to know $Q$ to conduct a principal component analysis, which can be estimated by the \edit{techniques} described in Section \ref{Sec: AbstractLimit theorems, fda case}. 

If we know $Q$ and if $e_1,...,e_d$ denote the first $d$ eigenfunctions of $Q$ corresponding to the largest $d$ eigenvalues $\lambda_1,...,\lambda_d$, we find with the notation $P_d f:= \sum_{i=1}^d \langle f,e_d\rangle e_d$ that 
$$\mathbb E[\|W_t^Q-P_dW_t^Q\|^2]=t \sum_{i=d+1}^{\infty} \lambda_i= t \text{tr}((I-P_d)Q(I-P_d)).$$
This mean squared error is minimal for all $d$-dimensional linear approximations of $W_t^Q$. 
 The good approximation property carries over to the process $f$ itself if we exchange $X$ by its approximation $X^d:=P_d X$ in the formulation of the dynamics \eqref{evol eq}. That is, if $f^d$ denotes the mild solution to the evolution equation
\eqref{eq: low dim evol eq}
and $f'_t= \int_0^t \mathcal S(t-s) dW_t^Q$ and ${f'_t}^d= \int_0^t \mathcal S(t-s) d(P_dW^Q_t)$ are the continuous parts (which are here  interpreted to be cleaned from outliers) of  the processes $f$ and $f^d$,  we have 
\begin{align*}
\mathbb E\left[\frac 1T\int_0^T\|f_t'-{f'_t}^d\|^2dt\right]
=& \frac 1T\int_0^T \text{tr}( \mathcal S(t-s)(I-P_d)Q (I-P_d) \mathcal S(t-s)^* ) ds\\
\leq &   \sup_{t\in [0,T]}\||\mathcal S(t)\|_{\text{op}}^2 \text{tr}((I-P_d)Q (I-P_d))
\end{align*}
Hence, the approximation error for the continuous part of the solution is proportional to the optimal approximation error for $W$.
In particular,
if we even have a semigroup that is contractive, the mean squared error for the process $f$ is at least as good as the mean squared error for approximating $W_t$ by its linear optimal approximation $(W_t)^d:=\sum_{i=1}^d\langle W_t e_i\rangle e_i$.  The next example promises that this approximation method can improve the fit of the model significantly, when compared to a covariance based dimension reduction.
\begin{example}[$X$ can be one-dimensional while the covariance is infinite dimensional]
Assume that $f$ is stationary with covariance $\mathcal C=\int_0^{\infty} \mathcal S(s) Q\mathcal S(s)ds$, where we assume the integral to converge.  For instance,  let $H=L^2(0,1)$, $\mathcal S(t) f(x)=f(x+t)\indicator_{[0,1]}(x+t)$ define the nilpotent left-shift semigroup and $Q=\indicator_{[0,1]}^{\otimes 2}$.  Then, by definition $Q$ is one-dimensional,  while the integral kernel $c$ corresponding to $\mathcal C$ is given by $c(x,y)=(1-\max(x,y))$, which defines an infinite-dimensional kernel with Mercer decomposition 
$$c(x,y)= \sum_{i=1}^{\infty} \frac{1}{(i-1/2)^2\pi^2} \sqrt 2 \sin(\pi(k-1/2)(1-x))sin(\pi(k-1/2)(1-y)).$$
Thus,  the covariance is infinite-dimensional. To explain e.g. $95\%$ of the variation of the data, we need at least 5 factors when dimension reduction is conducted on the state space of $f$ via the covariance,  while only one random driver is needed to  describe $X$ perfectly.
\end{example}

A further advantage of the dimension reduction of $X$ instead of $f$ is that we do not need to assume $f$ to be stationary. However, 
the dynamically consistent dimension reduction technique and also other techniques for functional data can be applied in more general situations than the case in which the continuous driver is a $Q$-Wiener process.  Some comments are in order.
\begin{remark}[Generalization to more general driving semimartingales]
A generalization of this procedure under more general assumptions is obvious,  exchanging the role of the covariance $Q$ of $W^Q$ with the quadratic variation of the continuous part of the driving semimartingale.  For estimation, we can can apply Theorems \ref{T: Identification of continuous and disc quadr var}, \ref{T: rates of convergence} and \ref{T: CLT without disc}.  Here, not even moment assumptions on the coefficients would be necessary. However, 
the derived random factors could just be evaluated a posteriori and would need to be considered time-varying a priori. an application of this can be found in \cite{Schroers2024b}.
If $T$ is large,  the results of Theorem \ref{T: Long-time no disc} are helpful, since they again enable the estimation of a stationary covariance operator of the continuous driving martingale,  yielding a procedure that provides a time-independent factor structure.
\end{remark}

\subsection{Estimation of a HEDIH Model}\label{Sec: Estimation of a Heidih model}

Let us now outline how Theorem \ref{T: Long-time no disc} can be used for the estimation of a simple infinite-dimensional stochastic volatility model.
\edit{Precisely, we consider a class of volatility processes of the form $\Sigma_t=Y_t^{\otimes 2}$ in which the process $Y$ follows the dynamics}
\begin{equation}\label{HEIDIH model}
dY_t=-\mathfrak C Y_t dt + d\mathcal W_t,
\end{equation}
where $\mathcal W$ is a cylindrical Wiener process and $\mathfrak C$ is a symmetric,  positive, unbounded operator with orthonormal eigenvectors  $(e_i)_{i\in \mathbb N}$ and eigenvalues $(\lambda_i)_{i\in \mathbb N}$ in decreasing order, such that
$$\sum_{i=1}^{\infty}\lambda_i^{-1}<\infty.$$
\edit{This framework generalizes the HEIDIH model from \cite{Petersson2022}.}

The operator $-\mathfrak C$ with domain $D(\mathfrak C)=\{h\in H: \sum_{i=1}^{\infty} \lambda_i^2\langle h, e_i\rangle^2\}$ is by the Hille-Yosida theorem the generator of the strongly continuous semigroup 
$\mathfrak S(t)=\sum_{i=1}^{\infty} e^{-\lambda_i t} e_i^{\otimes 2}.$
Assuming moreover, that $Y_0\sim N(0, \mathfrak C^{-1}/2)$, then $Y$ is a stationary process and can be written as the stochastic convolution $Y_t=\int_{-\infty}^t \mathfrak S(t-s) d\mathcal W_s$ for $t\geq 0$. Obviously $\Sigma_t$ has mean 
$\mathcal C= \mathfrak C^{-1}/2.$
Some simple calculations show that
$$\mathbb E\left[\left\|\frac 1T\int_0^T Y_s^{\otimes 2}ds-\mathfrak C^{-1}/2\right\|_{\text{HS}}^2\right]= \frac 1{T^2}\sum_{i=1}^{\infty}\int_0^T\int_0^T \left(\int_{-\infty}^s e^{-\lambda_i (t+s-2u)}du\right)^2 ds dt \leq\frac{ \| \mathfrak C^{-1}\|_{HS}^2}{T}$$
\edit{ such that $a_T=T^{-\frac 12}$ for $a_T$  specified in Theorem \ref{T: Long-time no disc}. Moreover,  }
      \begin{align*}
      \sup_{r >0}   \mathbb E\left[\left\|\frac{(I-\mathcal S(r))\Sigma_s}{r^{\gamma}}\right\|_{\text{HS}}\right]\leq  & \mathbb E[\|Y_s\|^2]^{\frac 12}\left(\sup_{r>0} \frac{\mathbb E[\|(I-\mathcal S(r))Y_s\|^2]}{r^{2\gamma}}\right)^{\frac 12}\\
      \leq &  \text{tr}(\edit{\mathfrak C^{-1}})^{\frac 12}\sup_{r>0} \left(\frac{\text{tr}((I-\mathcal S(r))\edit{\mathfrak C^{-1}}(I-S(r))^*)}{r^{2\gamma}}\right)^{\frac 12}.
    \end{align*}
    Assumption \ref{In proof II}($\gamma$) is now satisfied if 
    \begin{equation}\label{eq: regularity condition for the HEIDIH}
    \sup_{r>0}\text{tr}\left((I-\mathcal S(r))\edit{\mathfrak C^{-1}}(I-S(r))^*\right)/r^{2\gamma}<\infty.
    \end{equation} 
\edit{Using  Theorem \ref{T: Long-time no disc},} this leads to 
\begin{lemma}
If \eqref{eq: regularity condition for the HEIDIH} holds for some $\gamma\in (0,1/2]$, then in probability and for $T,n\to\infty$
\begin{equation}\label{eq: Consistent estimation of covariance in HEIDIH model}
\frac 2T SARCV_T^n - \mathfrak C^{-1} = \mathcal O_p(T^{-\frac 12}+\Delta_n^{\gamma}).
\end{equation}
\end{lemma}

\edit{Let us now check how fast convergence of the estimator $\mathcal C$ is in the context of} the \edit{concrete }setting of  \cite{Petersson2022}  which is designed to model energy forward price evolutions and is a special case of the model \eqref{HEIDIH model}. We set $H=L^2(0,1)$ and $\mathfrak C=\eta \partial_{x,x}$ with Neumann boundary conditions.  Moreover,  \edit{in \cite{Petersson2022}},  the operator governing the dynamics of $f$ is the derivative $\mathcal A=\partial_x$ with domain $D(\mathcal A)=\{f: f\text{ is absolutely contninuous and }f(1)=0\}$, which generates the nilpotent left-shift semigroup on $L^2(0,1)$ given as $\mathcal S(t)f(x)=f(x+t)\indicator_{[0,1]}(x+t)$ for $x,t\geq 0$ and $f\in L^2(0,1)$ (c.f.  \cite{Engel1999}). \edit{ Observe that $\mathfrak C$ and $\mathcal A$ are two distinct unbounded operators on $H$.}

 \edit{Since $H=L^2(0,1)$,} $\mathfrak C^{-1}$ is an integral kernel operator
$$\mathfrak C^{-1} f(x)= \int_0^1 \mathfrak c(x,y) f(y) dy.$$
\edit{ Some straightforward calculations prove that also $(I-\mathcal S(r))\mathfrak C^{-1}(I-S(r))^*)$ is an integral kernel operator with kernel $\tilde c$ given for all $x,y\in [0,1]$ by
$$\tilde c (x,y)= c(x,y)-\indicator_{[0,1-r]}(y) c(x,y+r)-\indicator_{[0,1-r]}c(x+r,y)+\indicator_{[0,1-r]}\indicator_{[0,1-r]}(y)c(x+r,y+r).$$ 
Consequently 
\begin{align*}
\text{tr}((I-\mathcal S(r))\mathfrak C^{-1}(I-S(r))^*)
= & \int_0^{1-r} c(x,x) -2 c(x+r)+c(x+r,x+r)dx +\int_{1-r}^1 c(x,x) dx.
\end{align*}
}
Now, if $\mathfrak c$ is H{\"o}lder continuous such that
$
\sup_{r>0, x\in [0,1]}(\mathfrak c(x+r,x)-\mathfrak c(x,x))/r^{2\gamma}<\infty$
then Condition \eqref{eq: regularity condition for the HEIDIH}($\gamma$) is satisfied. 
It is $\mathfrak c(x,y)=\eta (min(x,y)-xy)$ (c.f.  \cite{cavoretto2015}) which is Lipschitz continuous,  so \eqref{eq: Consistent estimation of covariance in HEIDIH model} holds with $\gamma= 1/2$.  

It can be speculated that similar techniques  for the estimation of the long-time estimator $SARCV_T^n(u_n,-)/T$ can be found for other functionals of the distribution of these models apart from their mean, enabling the estimation of more complex infinite-dimensional stochastic volatility models.

\edit{
\section{Simulation Evidence}\label{Sec: Simulation Study}
We conduct a simulation study to examine the impact of semigroup adjustment on measuring quadratic covariations, truncating jumps, and dynamically consistent principal component analysis. For this purpose, we derive discrete samples from the stochastic PDE
$$
df_t = \partial_x f_t\,dt + dX_t, \qquad x \geq 0,
$$
on the Hilbert space $H = L^2(\mathbb R_+)$. The semigroup is given by left-shifts,
$
\mathcal S(t)f(x) = f(x+t), $ when $ x,t \geq 0.
$
The driving semimartingale $X_t$ is the sum of a $Q$-Wiener process $W^Q$ and a compound Poisson process $J$.  
For the covariance operator $Q$ we consider two scenarios: a rough scenario with Laplace kernel and a smooth scenario with Gaussian kernel,
$$
q(x,y) = e^{-|x-y|}, \qquad \forall x,y \in [0,2]\qquad \text{(rough scenario)},
$$
$$
q(x,y) = e^{-(x-y)^2}, \qquad \forall x,y \in [0,2]\qquad\text{(smooth scenario)}.
$$
We also consider two scenarios for $J$. In the first, no jumps occur and hence $J \equiv 0$. In the second, jumps occur according to a Poisson process with intensity $\lambda=2$, so that on average two jumps occur on $[0,1]$. Each jump corresponds to a spatially constant level shift on the interval $[0,2]$,
$$
\chi_i(x) = \zeta_i \,\mathbf{1}_{[0,2]}(x), \qquad \zeta_i \sim N(0,0.1).
$$
The restriction to a compact interval in the definitions of $q$ and $\chi_i$ are to ensure that all elements are elements well-defined in the $L^2(\mathbb R_+)$-framework.  Although below we evaluate the equation only for $x\leq 1$,  the particular bound at $2$ is chosen to avoid boundary issues in the simulation of the data.
}

\edit{
We simulate on a uniform grid and obtain
$$ 
f_{\frac in}\left(\frac jn\right), \qquad i,j = 1,\dots,n+1, \qquad n=100.
$$
These samples can be generated without approximation error; the implementation details are outlined in Appendix \ref{Sec: description of simulation scheme}.  
From the simulated field we construct discrete increments. For each $i,j=1,\dots,n$, we define the semigroup-adjusted and the regular increments as
$$
\widetilde{\mathbf{d_if}}_j = \tilde\Delta_i^n f\left(\frac jn\right) = f_{\frac{(i+1)}n}\left(\frac jn\right) - f_{\frac in}\left(\frac{(j+1)}n\right),
\,\,\,
\mathbf{d_if}_j = \Delta_i^n f\left(\frac jn\right) = f_{\frac{(i+1)}n}\left(\frac jn\right) - f_{\frac in}\left(\frac{j}n\right).
$$
Based on these increments we form the empirical covariation matrices
$$
sarcv = \sum_{i=2}^n \widetilde{\mathbf{d_if}} \,\widetilde{\mathbf{d_if}}^{\top}, 
\qquad 
rcv = \sum_{i=2}^n \mathbf{d_if}\,\mathbf{d_if}^{\top}.
$$
These estimators will be evaluated as estimators of the dicretized and restricted (to $[0,1]$) quadratic variation kernel $\mathbf q_{j,j'}$ with entries
$$\mathbf q_{j,j'} := q\left(\frac{j-j'}n\right)\qquad j,j'=1,...,n.$$ 
}

\edit{
When jumps occur, increments contaminated by jumps must be removed. We employ a truncation rule adapted from \cite{Schroers2024b}, which we summarize here.  
\begin{enumerate}
\item Sort the $25\%$ largest increments (in Euclidean norm) and temporarily exclude them.  
\item From the remaining increments, construct a preliminary covariance estimator $\mathbf q_{\text{prel}}$ by computing either $sarcv$ or $rcv$, and rescale to account for the reduced sample size.  
\item Compute an approximate Mahalanobis-type distance for each increment. Choose $d$ so that the leading $d$ eigenvalues $\lambda_1,\dots,\lambda_d$ of $\mathbf q_{\text{prel}}$ explain $90\%$ of the variation. For the corresponding eigenvectors $e_1,\dots,e_d$, define
$$
g(f)^2 := \sum_{i=1}^d \lambda_i^{-1} \langle e_i,f\rangle_{\mathbb R^n}^2 
+ \frac{\sum_{i=d+1}^n \langle e_i,f\rangle_{\mathbb R^n}^2}{\sum_{i=d+1}^n \lambda_i}.
$$
\item Classify increments as outliers if $g(f) > 3 \Delta^{0.49}\sqrt{d+1}$.  
\end{enumerate}
The threshold is motivated by the fact that, under Gaussian dynamics without jumps, $g(f)$ is close to a $\chi^2_d$ distribution. For example, with $d=5$, the chosen threshold excludes fewer than $0.1\%$ of genuine increments.  
The adjusted covariation matrices in the presence of jumps are then defined as
$$
sarcv_- = \sum_{i=2}^n \widetilde{\mathbf{d_if}} \,\widetilde{\mathbf{d_if}}^{\top} \,\mathbf 1_{\{g(\widetilde{\mathbf{d_if}})\leq 3 \Delta^{0.49}\sqrt{d+1}\}},
\qquad
rcv_- = \sum_{i=2}^n \mathbf{d_if}\,\mathbf{d_if}^{\top} \,\mathbf 1_{\{g(\mathbf{d_if})\leq 3 \Delta^{0.49}\sqrt{d+1}\}}.
$$
}

\edit{
To assess the estimators, we use two criteria.  
\begin{enumerate}
\item \textbf{Relative Hilbert-Schmidt error.}  
For two matrices $A,B \in \mathbb R^{n\times n}$ we define
$$
rel_err = \frac{\|\hat{\mathbf q}-\mathbf q\|_{\text{Frob}}}{\|\mathbf q\|_{\text{Frob}}},
$$
where $\|\cdot\|_{\text{Frob}}$ denotes the Frobenius norm. Here $\hat{\mathbf q}$ is either $sarcv, rcv, sarcv_-,$ or $rcv_-$.  
\item \textbf{Explained variance dimension.}  
For principal component analysis we define
$$
d_{0.95} := \min\Bigl\{d \in \{1,\dots,n\} : \frac{\sum_{i=1}^d \lambda_i}{\sum_{i=1}^n \lambda_i} \geq 0.95\Bigr\},
$$
where $\lambda_1 \geq \cdots \geq \lambda_n$ are respectively the eigenvalues of the symmetric positive matrices  $sarcv, rcv, sarcv_-,$ or $rcv_-$. This is the number of principal components needed to explain $95\%$ of the variation in the incremetns.
\end{enumerate}
}

\edit{
\begin{table}[ht]
\centering
\small
\begin{tabular}{llcccc}
\hline
\textbf{Kernel} & \textbf{Metric} & \textbf{sarcv} & \textbf{rcv} & \textbf{sarcv\_} & \textbf{rcv\_} \\
\hline
\multirow{2}{*}{$e^{-(x-y)^2}$} 
    & rel\_err     &  0.11 (0.07, 0.17) & 0.11 (0.07, 0.17) &  0.13 (0.08,0.20) & 0.13 (0.08,0.20) \\
    & $d_{0.95}$   & 2 (2,2) & 2 (2,2) & 2 (2,2) & 2 (2,2) \\
\multirow{2}{*}{$e^{-|x-y|}$} 
    & rel\_err     & 0.13 (0.10,0.18) & 0.30 (0.28,0,33) & 0.15 (0.11,0,22) & 0.53 (0.32,1.18) \\
    & $d_{0.95}$   & 5 (5,5) & 48 (47,50) & 5 (4,5) & 47 (42,48) \\
\hline
\end{tabular}
\caption{Medians and in brackets the $25\%$ and the $75\%$ quantiles of relative errors and explained variance dimensions for different estimators and noise kernels based on 10,000 Monte Carlo runs. }
\label{Tab1}
\end{table}
Table~\ref{Tab1} reports the results of a simulation study with 10,000 runs.  For smooth kernels, semigroup adjustment has little impact on approximation error or the effective dimension (true dimension $=2$). In contrast, for rough kernels, adjustment significantly improves both metrics: the relative error of $rcv$ is more than twice as large as that of $sarcv$, and the estimated effective dimension is drastically inflated (47-48 instead of the true value 5).  
When jumps are present, $sarcv_-$ performs in the case with jumps comparably to $sarcv$ in the case without jumps,  while $rcv_-$ deteriorates in the case with jumps relative to $rcv$ in the case without jumps. This indicates that detecting and truncating jumps is substantially easier with adjusted increments than with regular increments.  
}

\edit{
In summary,  the simulation evidence demonstrates that semigroup-adjusted increments provide clear benefits when volatility innovations are rough, both for quadratic variation estimation and for principal component analysis. In the smooth case, adjustment has little effect. With jumps, the proposed truncation rule is effective in conjunction with adjusted increments but less reliable when applied to regular increments.
}

\subsection* {Acknowledgments}
I thank Dominik Liebl, Fred Espen Benth, Alois Kneip and Andreas Petersson as well as two reviewers
for helpful comments and discussions. 

\subsection*{funding}
Funding by the 
Argelander program of the University of Bonn is gratefully acknowledged.

\subsection*{Supplementary Material}

The replication code for the simulation studys in Section \ref{Sec: Simulation Study} can be downloaded from:
\url{https://github.com/dschroers/SimFDASEE}

\bibliographystyle{plain}
\bibliography{bibliography}

\begin{thebibliography}{10}

\bibitem{Ait-SahaliaJacod2014}
Y.~A\"{i}t-Sahalia and J.~Jacod.
\newblock {\em High-Frequency Financial Econometrics}.
\newblock Princeton University Press, Princeton, New Jersey, 2014.

\bibitem{altmeyer2021}
R.~Altmeyer and M.~Rei{\ss}.
\newblock Nonparametric estimation for linear spdes from local measurements.
\newblock {\em Ann. Appl. Probab.}, 31(1):1--38, 2021.

\bibitem{Petersson2022}
F.~Benth, G.~Lord, and A.~Petersson.
\newblock The heat modulated infinite dimensional {H}eston model and its
  numerical approximation.
\newblock {\em Stochastics}, pages 1--41, 2022.

\bibitem{BenthRudigerSuss2018}
F.~Benth, B.~R{\"u}diger, and A.~S{\"u}ss.
\newblock Ornstein--{U}hlenbeck processes in {H}ilbert space with
  non-{G}aussian stochastic volatility.
\newblock {\em Stoch. Proc. Applic.}, 128(2):461--486, 2018.

\bibitem{Benth2022}
F.~Benth, D.~Schroers, and A.~Veraart.
\newblock A weak law of large numbers for realised covariation in a {H}ilbert
  space setting.
\newblock {\em Stoch. Proc. Applic.}, 145:241--268, 2022.

\bibitem{BSV2022}
F.~Benth, D.~Schroers, and A.~Veraart.
\newblock A feasible central limit theorem for realised covariation of spdes in
  the context of functional data.
\newblock {\em Ann. Appl. Probab.}, 34(2):2208--2242, 2024.

\bibitem{BenthSgarra2021}
F.~Benth and C.~Sgarra.
\newblock A {B}arndorff-{N}ielsen and {S}hephard model with leverage in
  {H}ilbert space for commodity forward markets.
\newblock {\em To appear in Finance Stoch.}, 2021.

\bibitem{BenthSimonsen2018}
F.~Benth and I.~Simonsen.
\newblock The {H}eston stochastic volatility model in {H}ilbert space.
\newblock {\em Stoch. Analysis Applic.}, 36(4):733--750, 2018.

\bibitem{Bibinger2020}
M.~Bibinger and M.~Trabs.
\newblock Volatility estimation for stochastic pdes using high-frequency
  observations.
\newblock {\em Stoch. Proc. Applic.}, 130(5):3005 -- 3052, 2020.

\bibitem{bjork1999}
T.~Bj{\"o}rk and B.~Christensen.
\newblock Interest rate dynamics and consistent forward rate curves.
\newblock {\em Math. Finance}, 9(4):323--348, 1999.

\bibitem{bjork2001}
T.~Bj{\"o}rk and L.~Svensson.
\newblock On the existence of finite-dimensional realizations for nonlinear
  forward rate models.
\newblock {\em Math. Finance}, 11(2):205--243, 2001.

\bibitem{cavoretto2015}
R.~Cavoretto, G.~Fasshauer, and M.~McCourt.
\newblock An introduction to the hilbert-schmidt svd using iterated brownian
  bridge kernels.
\newblock {\em Numer. Algorithms}, 68:393--422, 2015.

\bibitem{Chong2020}
C.~Chong.
\newblock High-frequency analysis of parabolic stochastic pdes.
\newblock {\em Ann. Statist.}, 48(2):1143--1167, 04 2020.

\bibitem{ChongDalang2020}
C.~Chong and R.~Dalang.
\newblock Power variations in fractional {S}obolev spaces for a class of
  parabolic stochastic {PDE}s.
\newblock {\em Bernoulli}, 29(3):1792 -- 1820, 2020.

\bibitem{Cialenco2018}
I.~Cialenco.
\newblock Statistical inference for {SPDE}s: an overview.
\newblock {\em Statist. Inf. Stoch. Proc.}, 20(2):309--329, 12 2018.

\bibitem{cialenco2022}
I.~Cialenco and H.-J. Kim.
\newblock Parameter estimation for discretely sampled stochastic heat equation
  driven by space-only noise.
\newblock {\em Stoch. Proc. Applic.}, 143:1--30, 2022.

\bibitem{cox2023}
S.~Cox, C.~Cuchiero, and A.~Khedher.
\newblock Infinite-dimensional {W}ishart-processes.
\newblock {\em arXiv preprint arXiv:2304.03490}, 2023.

\bibitem{Cox2021}
S.~Cox, S.~Karbach, and A.~Khedher.
\newblock An infinite-dimensional affine stochastic volatility model.
\newblock {\em Math. Finance}, 2021.

\bibitem{Cox2020}
S.~Cox, S.~Karbach, and A.~Khedher.
\newblock Affine pure-jump processes on positive {H}ilbert-{S}chmidt operators.
\newblock {\em Stoch. Proc. Applic.}, 151:191--229, 2022.

\bibitem{DPZ2014}
G.~Da~Prato and J.~Zabczyk.
\newblock {\em Stochastic Equations in Infinite Dimensions}, volume 152 of {\em
  Encyclopedia of Mathematics and its Applications}.
\newblock Cambridge University Press, Cambridge, second edition, 2014.

\bibitem{Engel1999}
K.~Engel and R.~Nagel.
\newblock {\em One-Parameter Semigroups for Linear Evolution Equations}, volume
  194.
\newblock Springer Science \& Business Media, 1999.

\bibitem{filipovic1999}
D.~Filipovi{\'c}.
\newblock A note on the {N}elson--{S}iegel family.
\newblock {\em Math. finance}, 9(4):349--359, 1999.

\bibitem{filipovic2000b}
D.~Filipovi{\'c}.
\newblock Exponential-polynomial families and the term structure of interest
  rates.
\newblock {\em Bernoulli}, pages 1081--1107, 2000.

\bibitem{filipovic2000c}
D.~Filipovi{\'c}.
\newblock Invariant manifolds for weak solutions to stochastic equations.
\newblock {\em Probab. Theory Related Fields}, 118(3):323--341, 2000.

\bibitem{Filipovic2000}
D.~Filipovi{\'c}.
\newblock {\em Consistency Problems for HJM Interest Rate Models}, volume 1760
  of {\em Lecture Notes in Mathematics}.
\newblock Springer, Berlin, 2001.

\bibitem{FilipovicTappeTeichmann2010}
D.~Filipovi{\'c}, S.~Tappe, and J.~Teichmann.
\newblock Term structure models driven by {W}iener processes and {P}oisson
  measures: existence and positivity.
\newblock {\em SIAM J. Fin. Math.}, 1(1):523--554, 2010.

\bibitem{filipovic2003}
D.~Filipovi{\'c} and J.~Teichmann.
\newblock Existence of invariant manifolds for stochastic equations in infinite
  dimension.
\newblock {\em J. Funct. Anal.}, 197(2):398--432, 2003.

\bibitem{Hildebrandt2021}
F.~Hildebrandt and M.~Trabs.
\newblock {Parameter estimation for SPDEs based on discrete observations in
  time and space}.
\newblock {\em Electron. J. Stat.}, 15(1):2716 -- 2776, 2021.

\bibitem{Hildebrandt2023}
F.~Hildebrandt and M.~Trabs.
\newblock Nonparametric calibration for stochastic reaction–diffusion
  equations based on discrete observations.
\newblock {\em Stoch.Proc.Applic.}, 162:171--217, 2023.

\bibitem{Jacod2008}
J.~Jacod.
\newblock Asymptotic properties of realized power variations and related
  functionals of semimartingales.
\newblock {\em Stoch. Proc. Applic.}, 118(4):517--559, 2008.

\bibitem{JacodProtter2012}
J.~Jacod and P.~Protter.
\newblock {\em Discretization of Processes}, volume~67 of {\em Stochastic
  Modelling and Applied Probability}.
\newblock Springer, Heidelberg, 2012.

\bibitem{JacodShirayev2003}
J.~Jacod and A.~Shirayev.
\newblock {\em Limit Theorems for Stochastic Processes}.
\newblock Grundlehren der mathematischen Wissenschaften. Springer Berlin,
  Heidelberg, 2003.

\bibitem{Knoche05}
C.~Knoche.
\newblock {\em Mild solutions of SPDEs driven by Poisson noise in infinite
  dimensions and their dependence on initial conditions}.
\newblock PhD thesis. University of Bielefeld, 2005.

\bibitem{Koski85}
T.~Koski and W.~Loges.
\newblock Asymptotic statistical inference for a stochastic heat flow problem.
\newblock {\em Statist. Probab. Lett.}, 3:185 -- 189, 1985.

\bibitem{Kruse2014}
R.~Kruse.
\newblock {\em Strong and Weak Approximation of Semilinear Stochastic Evolution
  Equations}, volume 2093 of {\em Lecture Notes in Mathematics}.
\newblock Springer International Publishing Switzerland,, 2014.

\bibitem{liu2015}
W.~Liu and M.~R{\"o}ckner.
\newblock {\em Stochastic partial differential equations: an introduction}.
\newblock Springer, 2015.

\bibitem{Mancini2009}
C.~Mancini.
\newblock Nonparametric threshold estimation for models with stochastic
  diffusion coefficient and jumps.
\newblock {\em Scand. J. Statist.}, 36:270--296, 2009.

\bibitem{GM2011}
V.~Mandrekar and L.~Gawarecki.
\newblock {\em Stochastic Differential Equations in Infinite Dimensions}.
\newblock Probability and Its Applications. Springer, Berlin, Heidelberg, 2011.

\bibitem{mandrekar2015}
V.~Mandrekar and B.~R{\"u}diger.
\newblock {\em {Stochastic integration in Banach spaces}}, volume~73 of {\em
  Probability Theory and Stochastic Modelling}.
\newblock Springer, 2015.

\bibitem{Marinelli2016}
C.~Marinelli and M.~Röckner.
\newblock On the maximal inequalities of {B}urkholder, {D}avis and {G}undy.
\newblock {\em Expo. Math.}, 34:1--26, 2016.

\bibitem{Panaretos2019}
V.~Masarotto, V.~M. Panaretos, and Y.~Zemel.
\newblock Procrustes metrics on covariance operators and optimal transportation
  of gaussian processes.
\newblock {\em Sankhya A}, 81(1):172--213, 2019.

\bibitem{Pasemann2021}
G.~Pasemann, S.~Flemming, S.~Alonso, C.~Beta, and W.~Stannat.
\newblock Diffusivity estimation for activator–inhibitor models: Theory and
  application to intracellular dynamics of the actin cytoskeleton.
\newblock {\em J. Nonlinear Sci.}, 31, 2021.

\bibitem{PZ2007}
S.~Peszat and J.~Zabczyk.
\newblock {\em Stochastic Partial Differential Equations with {L}\'{e}vy
  Noise}, volume 113 of {\em Encyclopedia of Mathematics and its Applications}.
\newblock Cambridge University Press, Cambridge, 2007.

\bibitem{Rudin1976}
W.~Rudin.
\newblock {\em Principles of Mathematical Analysis}.
\newblock McGraw-Hill, New York, 3 edition, 1976.

\bibitem{Schroers2024b}
D.~Schroers.
\newblock Dynamically consistent analysis of realized covariations in term
  structure models.
\newblock {\em Math. Finance (to appear)}, 2025.

\end{thebibliography}
\begin{appendix}

\section{Proofs of the Abstract Results}\label{Sec: Proofs}

In this section we prove Theorems \ref{T: Abstract identification of quadr var},  \ref{T: Identification of continuous and disc quadr var}, \ref{T: rates of convergence}, \ref{T: CLT without disc},  \ref{T: Long-time no disc}, and \edit{Corollary} \ref{C: Results for realized variation}.
For that, recall that
$f$ is
a c{\`a}dl{\`a}g adapted stochastic process defined on a filtered probability space $(\Omega,\mathcal F,(\mathcal F_t)_{t\in \mathbb R_+},\mathbb P)$ with right-continuous filtration $(\mathcal F_t)_{t\in \mathbb R_+}$ taking values in a separable Hilbert space $H$ and is given by
\begin{align*}
    f_t= &\mathcal S(t)f_0+\int_0^t\mathcal S(t-s)\alpha_s ds+\int_0^t \mathcal S(t-s)\sigma_s dW_s +\int_0^t \int_{H\setminus \{0\}} \mathcal S(t-s)\gamma_s(z) (N-\nu)(dz,ds),
\end{align*}
where $\alpha,$ $\sigma$, $\gamma$, $W$ and $N$, $\mathcal S$ (with generator $\mathcal A$) are defined as in section \ref{Sec: mild ito semimartingales}.
To make the structure of the proofs more consistent, some other auxiliary technical results that we use throughout are relegated to Section \ref{Sec: Auxiliary Technical Results}.

\subsection{Proofs of Theorem \ref{T: Abstract identification of quadr var}}
In this section we are going to prove  Theorem \ref{T: Abstract identification of quadr var}.

\begin{proof}[Proof of Theorem \ref{T: Abstract identification of quadr var}]
By Lemma \ref{L: Localisation} we can assume the stronger Assumption \ref{As: localised integrals}, that is, that there is a constant $A>0$ such that $$\int_0^{\infty} \|\alpha_s\| ds+\int_0^{\infty} \|\sigma_s\|_{\text{HS}}^2ds+\int_0^{\infty} \int_{H\setminus \{0\}} \|\gamma_s(z)\|^2 F(dz) ds<A.$$
Let $(e_j)_{j\in \mathbb N}$ be an orthonormal basis of $H$ such that $e_j\in D(\mathcal A^*)$ for all $j\in \mathbb N$ contained in the domain $D(\mathcal A^*)$ of the generator $\mathcal A^*$ of the adjoint semigroup $(\mathcal S(t)^*)_{t\geq 0}$, which always is a semigroup on $H$  (c.f. \cite[p.44]{Engel1999}).
Then,  \edit{for} all $k,l\in\mathbb N$ we have by Lemma \ref{L: Reduction to semimartingales is possible} that
\begin{align}\label{Reduction to semimartingale convergence for finite-dimensional distributions 0}
  &\left| \langle SARCV_t^n, e_k\otimes e_l\rangle_{L_{\text{HS}}(H)}-\langle [X,X]_t,e_k\otimes e_l
\rangle_{\text{HS}}\right|\\
  \leq &  \left| \sum_{i=1}^{\ul} \langle  \Delta_i^n X,e_k\rangle\langle\Delta_i^n X,e_l\rangle \indicator_{\|\Delta_i^n X\|\leq u_n}-\langle [X,X]_t,e_k\otimes e_l
\rangle_{\text{HS}}\right|+ \mathcal O_p\left(\Delta_n^{1-2w}\right)\notag
\end{align}
as the matrix valued process 
$\left(\langle [X,X]_t,e_k\otimes e_l
\rangle_{\text{HS}}\right)_{k,l=1,...,N,t\geq 0}$
corresponds to the quadratic variation of 
the $N$-dimensional It{\^o} semimartingale 
$(\langle X_t,e_1\rangle,...,\langle X_t,e_N\rangle)_{t\geq 0}.$ 
Defining by $P_N\in L(H)$ the projection onto $span(e_j:j=1,...,N)$ for $N\in \mathbb N$, we obtain by \eqref{Reduction to semimartingale convergence for finite-dimensional distributions 0} for all $N\in\mathbb N$ that as $n\to\infty$ it is
\begin{equation}
P_N SARCV_t^n P_N\overset{u.c.p.}{\longrightarrow}[P_N X,P_N X]_t.
\end{equation}
Moreover, by H{\"o}lder's inequality it is
\begin{align*}
& \mathbb E \left[\sup_{t\in [0,T]}\left\|SARCV- P_N SARCV P_N\right\|\right]\\
  \leq &2\left(\sum_{i=1}^{\ulT} \mathbb E\left[\left\|(P_N-I) \tilde \Delta_i^n f\right\|^2\right]\right)^{\frac 12}\left(\sum_{i=1}^{\ulT}\mathbb E\left[\left\| \tilde \Delta_i^n f\right\|^2\right]\right)^{\frac 12}.
\end{align*}
By \eqref{increment-vanishes asymptotically} in Lemma \ref{L: Projection convergese uniformly on the range of volatility}, the latter converges to  $0$ as $N\to \infty$ uniformly in $n$.  Since \begin{align*}
  & \sup_{t\in [0,T]} \left\|[P_N X,P_N X]-[ X, X]\right\|\\
   \leq &  \int_0^T \left\|P_N \Sigma_s P_N-\Sigma_s\right\| ds+ \sum_{s\leq T} \left\|P_N (X_s-X_{s-})^{\otimes 2}P_N-(X_s-X_{s-})^{\otimes 2}\right\|
\end{align*}
converges to $0$ by dominated convergence, since  for all $s\geq 0$  by Proposition 4 and Lemma 5 in \cite{Panaretos2019} that as $N\to\infty$ and almost surely
\begin{align*}
    \|P_N \Sigma_s P_N-\Sigma_s\|\to 0 \text{ and } \|P_N (X_s-X_{s-})^{\otimes 2}P_N-(X_s-X_{s-})^{\otimes 2}\|\to 0,
\end{align*} 
the proof is complete.
\end{proof}

\subsection{Proof of Theorem \ref{T: Identification of continuous and disc quadr var}}

We will now prove Theorem \ref{T: Identification of continuous and disc quadr var}.

\begin{proof}[Proof of Theorem \ref{T: Identification of continuous and disc quadr var}] The proof is similar to the one of Theorem \ref{T: Abstract identification of quadr var}.
By Lemma \ref{L: Localisation} we can assume the stronger Assumption \ref{As: SH}(2),  which essentially requires the coeffients to be bounded.
Define 
again,
$(e_j)_{j\in \mathbb N}$ to be an orthonormal basis of $H$ such that $e_j\in D(\mathcal A^*)$ for all $j\in \mathbb N$. 
For all $k,l\in\mathbb N$ we then have by Lemma \ref{L: Reduction to semimartingales is possible} that
\begin{align}\label{Reduction to semimartingale convergence for finite-dimensional distributions}
  &\left| \langle SARCV_t^n(u_n,-), e_k\otimes e_l\rangle_{\text{HS}}- \langle[X^C,X^C]_t, e_k\otimes e_l\rangle_{\text{HS}} \right|\\
  \leq & o_p(1)+ \left| \sum_{i=1}^{\ul} \langle  \Delta_i^n X,e_k\rangle\langle\Delta_i^n X,e_l\rangle \indicator_{\|\Delta_i^n X\|\leq u_n}-\langle[X^C,X^C]_t, e_k\otimes e_l\rangle_{\text{HS}}  \right|.\notag
\end{align}
The latter converges to $0$ uniformly on compacts in probability by Theorem 9.2.1 in \cite{JacodProtter2012}.  
Defining by $P_N\in L(H)$ the projection onto $span(e_j:j=1,...,N)$ for $N\in \mathbb N$, we obtain that \eqref{Reduction to semimartingale convergence for finite-dimensional distributions} implies that for all $N\in\mathbb N$ it is
$$P_N SARCV_t^n(u_n,-)P_n-[P_NX^C,P_NX^C]_t\overset{u.c.p.}{\longrightarrow} 0.$$
It is now enough to prove that as $N\to \infty$ we have
$$\sup_{n\in \mathbb N}(I_{L_{\text{HS}}(H)}-P_N\cdot P_N) \left(SARCV_t^n(u_n,-)-[X^C,X^C]_t\right)\overset{u.c.p.}{\longrightarrow} 0.$$
As in the proof of Theorem \ref{T: Abstract identification of quadr var} we obtain by H{\"o}lder's inequality that
\begin{align*}
& \mathbb E \left[\sup_{t\in [0,T]}\left\|SARCV_t^n(u_n,-)- P_N SARCV_t^n(u_n,-) P_N\right\|\right]\\
 \leq & \left(\sum_{i=1}^{\ulT} \mathbb E\left[\left\|(P_N-I) \tilde \Delta_i^n f\right\|^2\right]\right)^{\frac 12}\left(\sum_{i=1}^{\ulT}\mathbb E\left[\left\| \Delta_i^n f\right\|^2\right]\right)^{\frac 12}
\end{align*}
which converges to $0$ as $N\to \infty$ uniformly in $n$ by \eqref{increment-vanishes asymptotically} below. The convergence $\sup_{n\in \mathbb N}(I_{L_{\text{HS}}(H)}-P_N\cdot P_N) \int_0^t \Sigma_s ds$ holds uniformly on compacts in probability, which can be shown as in the proof of Theorem \ref{T: Abstract identification of quadr var}.
\end{proof}

\subsection{Proof of Theorem \ref{T: rates of convergence}}

The goal of this section is to prove the subsequent theorem which, in particular, shows Theorem \ref{T: rates of convergence}, since Assumption \ref{As: SH}(r), which is a localized version of Assumption \ref{As: H}, implies Assumption \ref{In proof: Very very Weak localised integrability Assumption on the moments}(p,r) for all $p> 0$ and then Lemma \ref{L: Localisation}(iii) applies.

\begin{theorem}\label{T: Uniform convergence theorem for realized variation}
        If Assumptions \ref{As: spatial regularity}($\gamma$) and \ref{In proof: Very very Weak localised integrability Assumption on the moments}(p,r) hold for $\gamma\in (0,1/2]$, $p>2/(1-2w)$ and for some $r\in (0,2)$ we have for $\rho<(2-r)w$
\begin{equation}\label{App: Rate of convergence}
   \sup_{t\in [0,T]}\left\|SARCV_t^n(u_n,-)-[X^C,X^C]_t\right\|_{\text{HS}}=\mathcal O_p\left( \Delta_n^{\min(\gamma,\rho)}\right).
\end{equation}
If, in addition, Assumption \ref{In proof II}($\gamma$) holds, we can find a constant $K>0$, which is independent of $T$ and $n$ such that
\begin{equation}\label{App: Rate of convergence2}
    \mathbb E\left[\sup_{t\in [0,T]}\left\|SARCV_t^n(u_n,-)-[X^C,X^C]_t\right\|_{\text{HS}}\right]\leq K T \Delta_n^{\min(\gamma,\rho)}.
\end{equation}
\end{theorem}

Theorem \ref{T: Uniform convergence theorem for realized variation} is proved via a discretization procedure in several steps. Before we introduce the scheme, 
let us introduce some important notation.
If Assumption \ref{As: H}(r) holds for $0<r\leq 1$ we can write 
\begin{align*}
    f_t= \mathcal S(t)f_0+\int_0^t \mathcal S(t-s)\alpha_s' ds+\int_0^t \mathcal S(t-s)\sigma_s dW_s+\int_0^t \int_{H\setminus \{0\}} \mathcal S(t-s)\gamma_s(z) N(dz,ds),
\end{align*}
where 
$$\alpha_s'= \alpha_s-\int_{H\setminus \{0\}} \gamma_s(z) F(dz)$$
and the integral w.r.t. the (not compensated) Poisson random measure $N$ is well defined (for the second term recall the definition of the integral e.g. from \cite[Section 8.7]{PZ2007}). 
In this case, we define $f=f'+f''$ where
\begin{align}\label{Eq: Continuous discontinuous condition for low r}
   & f_t':=\mathcal S(t) f_0+\int_0^t \mathcal S(t-s)\alpha_s'+\int_0^t \mathcal S(t-s)\sigma_sdW_s,\\
    &f_t'':=\int_0^t \int_{H\setminus \{ 0\}}\mathcal S(t-s)\gamma_s(z) N(dz,ds).\notag
\end{align}

In the case that Assumption \ref{As: H}(r) holds for $r\in (1,2)$ we define $f=f'+f''$ where
\begin{align}\label{Eq: Continuous discontinuous condition for large r}
   & f_t':=\mathcal S(t) f_0+\int_0^t \mathcal S(t-s)\alpha_s+\int_0^t \mathcal S(t-s)\sigma_sdW_s,\\
    &f_t'':=\int_0^t \int_{H\setminus \{ 0\}}\mathcal S(t-s)\gamma_s(z) (N-\nu)(dz,ds).\notag
\end{align}

Theorem \ref{T: Uniform convergence theorem for realized variation} is then shown in several steps in which we derive constants $K_1,K_2,K_3,K_4>0$ and real-valued sequences $(\phi^1_n)_{n\in \mathbb N}$, $(\phi^2_n)_{n\in \mathbb N}$ such that $\phi^1_n\to 0$ and $\phi_n^2\to 0$ as $n\to\infty$ such that for $\rho< (2-r)w$
\begin{itemize}
    \item[(i)]$
         \mathbb E\left[\sup_{t\in [0,T]}\left\|SARCV_t^n(u_n,-)-\sum_{i=1}^{\ul} (\tilde\Delta_i^n f')^{\otimes 2}\indicator_{g_n(\tilde\Delta_i^n f')\leq u_n}\right\|_{\text{HS}}\right]
         \leq  K_1 T\Delta_n^{\rho}\phi^1_n;$
    \item[(ii)]  $
         \mathbb E\left[ \sup_{t\in [0,T]}\|\sum_{i=1}^{\ul} ( \tilde\Delta_i^n f')^{\otimes 2}\indicator_{g_n(\tilde\Delta_i^n f')\leq u_n}-\sum_{i=1}^{\ul} (\tilde\Delta_i^n f')^{\otimes 2}\|_{\text{HS}}\right]
         \leq  K_2 T \Delta_n^{\frac 12}\phi^2_n;$
     \item[(iii)] Using the notation $\Sigma_s^{\mathcal S_n}:= \mathcal S(\lfloor s/\Delta_n\rfloor \Delta_n-s)\Sigma_s\mathcal S(\lfloor s/\Delta_n\rfloor \Delta_n-s)^*,$
     we have
     \begin{align*}
         \mathbb E\left[ \sup_{t\in [0,T]}\|\sum_{i=1}^{\ul} (\tilde\Delta_i^n f')^{\otimes 2}-\int_0^t \Sigma_s^{\mathcal S_n}  \|_{\text{HS}}\right]
         \leq  K_3 T  \Delta_n^{\frac 12};\end{align*}
     \item[(iv)] We have
     \begin{align*}
      \sqrt n \left\| \sup_{m\in \mathbb N\cup_{\infty}}\int_0^t \Sigma_s^{\mathcal S_n} ds-[X^C,X^C]_t \right\|_{\text{HS}}
        =\mathcal O_p\left(\Delta_n^{\gamma}\right)
        \end{align*}
     and if even Assumption \ref{In proof II}($\gamma$) holds it is
     \begin{align*}
         \mathbb E\left[\sqrt n \left\| \int_0^t \Sigma_s^{\mathcal S_n} ds-[X^C,X^C]_t \right\|_{\text{HS}}\right]
         \leq  K_4 T \Delta_n^{\gamma}.\end{align*}
\end{itemize}
Then Theorem \ref{T: Uniform convergence theorem for realized variation} follows with $K=K_1+K_2+K_3+K_4$. 
Therefore,  it remains to prove the estimates (i), (ii), (iii) and (iv).

We start with (i).
\begin{lemma}\label{Lem: Truncated minus continuous truncated is AN}
Let Assumption \ref{In proof: Very very Weak localised integrability Assumption on the moments}(p,r) and Assumption \ref{As: H}(r) hold for some $p>max(2/(1-2w),(2-2wr)/(2(2-r)w))$, $r\in (0,2)$.
Then we can find a constant $K>0$ and a sequence $(\phi_n)_{n\in \mathbb N}$ such that $\phi_n\to 0$ as $n\to\infty$ and 
    \begin{align*}
        & \mathbb E\left[\sup_{t\in [0,T]}\left\|\sum_{i=1}^{\ul} (\tilde\Delta_i^n f')^{\otimes 2}\indicator_{g_n(\tilde\Delta_i^n f')\leq u_n}-(\tilde\Delta_i^n f)^{\otimes 2}\indicator_{g_n(\tilde \Delta_i^n f)\leq u_n}\right\|_{\text{HS}}\right]\\
         \leq & K T \Delta_n^{\min\left(\frac{p-2}p,\frac{(2-r)(p-1)}{2p}, \frac{(1-wr)(p-1)}p-1+2w\right)} \phi_n .
    \end{align*}
    If we can even choose $p$ arbitrarily large, as under Assumption \ref{As: SH}(r), this yields that for all $\rho<(2-r)w$ we have
     \begin{align*}
         \mathbb E\left[\sup_{t\in [0,T]}\left\|\sum_{i=1}^{\ul} (\tilde\Delta_i^n f')^{\otimes 2}\indicator_{g_n(\tilde\Delta_i^n f')\leq u_n}-(\tilde\Delta_i^n f)^{\otimes 2}\indicator_{g_n(\tilde \Delta_i^n f)\leq u_n}\right\|_{\text{HS}}\right]
         \leq  K T \Delta_n^{ \rho } \phi_n .
    \end{align*}
\end{lemma}
\begin{proof}
Without loss of generality we may assume that in \eqref{abstract g conditions} it is $c=1$.
for a constant $C>0$. Denote by $f'$ and $f''$ the terms corresponding to the decompositions \eqref{Eq: Continuous discontinuous condition for low r} in the case that $r\leq 1$ and \eqref{Eq: Continuous discontinuous condition for large r} in the case $r\in (1,2]$.
First we decompose
\begin{align}\label{Eq: Jump Discretization}
  & \sum_{i=1}^{\ul} (\tilde\Delta_i^n f')^{\otimes 2}\indicator_{g_n(\tilde\Delta_i^n f')\leq u_n}-(\Delta_i^n f)^{\otimes 2}\indicator_{g_n(\tilde\Delta_i^n f)\leq u_n}\notag\\
    = & \sum_{i=1}^{\ul} \left((\tilde\Delta_i^n f')^{\otimes 2}-(\tilde\Delta_i^n f)^{\otimes 2}\right)\indicator_{g_n(\tilde\Delta_i^n f')\leq u_n,g_n(\tilde\Delta_i^n f)\leq u_n}+ \sum_{i=1}^{\ul} (\tilde\Delta_i^n f')^{\otimes 2}\indicator_{g_n(\tilde\Delta_i^n f')\leq u_n<g_n(\tilde\Delta_i^n f)}\notag\\
    &+ \sum_{i=1}^{\ul} (\tilde\Delta_i^n f)^{\otimes 2}\indicator_{g_n(\tilde\Delta_i^n f)\leq u_n<g_n(\tilde\Delta_i^n f')}\notag\\
    = & (i)_t^{n}+(ii)_t^{n}+(iii)_t^{n}.
\end{align}
We will now prove that  these summands have the required uniform bounds in  $t\in [0,T]$.

We start with  $(iii)_t^{n}$. For that, choose $1<\tilde q< p/2$ and $q=\tilde q/(\tilde q-1)$ and apply H{\"o}lder's inequality to obtain
\begin{align}\label{Est of martingale equals truncated martingale RV:I}
  \mathbb E\left[\sup_{t\in [0,T]}\| (iii)_t^{n}\|\right] 
\leq     \sum_{i=1}^{\ulT}  \mathbb E\left[\left\|\tilde\Delta_i^n f'\right\|^{2\tilde q}\right]^{\frac 1{\tilde q}}  \mathbb E\left[\indicator_{u_n<\|\tilde\Delta_i^n f'\|}^q\right]^{\frac 1q}.
\end{align}
 It is now by Lemma \ref{lem: Generalized BDG} (which applies as $\tilde q < p/2$) for a constant $K>0$ (which we will increase accordingly throughout the subsequent arguments)
\begin{equation}\label{Est of martingale equals truncated martingale RV:II}
\mathbb E\left[\left\|\tilde\Delta_i^n f'\right\|^{2\tilde q}\right]^{\frac 1{\tilde q}}  \leq\mathbb E\left[\left\|\tilde\Delta_i^n f'\right\|^{p}\right]^{\frac {2}{p}}  \leq  K\Delta_n.    
\end{equation}
Moreover, we have again by Lemma \ref{lem: Generalized BDG} (which we can apply since $2/(1-2w)
<p$) that
\begin{align}\label{Est of martingale equals truncated martingale RV:III}
        \mathbb E\left[\indicator_{u_n<\|\tilde\Delta_i^n f'\|}\right]
   \leq & \mathbb E\left[ \frac{\| \tilde\Delta_i^n f'\|^{2/(1-2w)}}{u_n^{2/(1-2w)}}\right]
    \leq K\Delta_n.
\end{align}
Combining \eqref{Est of martingale equals truncated martingale RV:I}, \eqref{Est of martingale equals truncated martingale RV:II} and \eqref{Est of martingale equals truncated martingale RV:III} we find a constant $K>0$ such that
\begin{align*}
    \mathbb E\left[\sup_{t\in [0,T]}\| (iii)_t^n\|\right]\leq  K \ulT \Delta_n \Delta_n^{\frac {1}q}.
    \leq KT\Delta_n^{\frac {p-2}p}.
\end{align*}

For the second summand,  since by the subadditivity of $g_n$ we have  $g_n(\tilde\Delta_i^n f)\leq (g_n(\tilde\Delta_i^n f')+g_n(\tilde\Delta_i^n f''))$ and, hence,
$2g(\tilde\Delta_i^n f')\leq u_n<g_n(\tilde\Delta_i^n f)$ implies that $g_n(\tilde\Delta_i^n f'')\geq u_n/2$, we find
\begin{align*}
 &\left\|\sqrt n\sum_{i=1}^{\ul} (\tilde\Delta_i^n f')^{\otimes 2}\indicator_{g_n(\tilde\Delta_i^n f')\leq u_n<g_n(\tilde\Delta_i^n f)} \right\|_{\text{HS}}\\
     \leq &   \sum_{i=1}^{\ul}  \left\|\tilde\Delta_i^n f' \right\|^2\indicator_{u_n/2\leq \|\tilde\Delta_i^n f'\|}+  \sum_{i=1}^{\ul}  \left\|\tilde\Delta_i^n f' \right\|^2\indicator_{u_n/2<\|\tilde\Delta_i^n f''\|}\\
   := & (ii)_t^{n}(a)+(ii)_t^{n}(b).
\end{align*}
For the first part, we can proceed as for $(iii)_t^{n}$ and obtain a constant $K>0$ and a sequence $(\phi_n)_{n\in \mathbb N}$ converging to $0$ such that
\begin{align*}
    \mathbb E\left[\sup_{t\in [0,T]}\| (ii)_t^n(a)\|\right]\leq KT\phi_n\Delta_n^{\frac{p-2}p}.
\end{align*}
For the second part, we can estimate 
\begin{align*}
  \|  (ii)_t^{n}(b)\|\leq  \sum_{i=1}^{\ul}  \left\|\tilde\Delta_i^n f' \right\|^2\left(1\wedge 2\frac{\|\tilde\Delta_i^n f''\|}{u_n}\right).
\end{align*}
Observe that $\|\tilde \Delta_i^n f''\|=\|\tilde \Delta_i^n f-\tilde \Delta_i^n f'\|\leq\|\tilde \Delta_i^n f\|+\|\tilde \Delta_i^n f'\| $. With $v_n:=u_n/\sqrt{\Delta_n}$ 
we also find
\begin{align*}
     \|(i)_t^n\|_{\text{HS}} 
     \leq &  \Delta_n\sum_{i=1}^{\ulT}\left(1+ \frac{\|\tilde\Delta_i^n f'\|}{\sqrt{\Delta_n}}\right)\left(\frac{\|\tilde\Delta_i^n f''\|}{\sqrt{\Delta_n}}\wedge 1+\left(\frac{\|\tilde\Delta_i^n f''\|}{\sqrt{\Delta_n}}\right)^2\wedge (2v_n)^2\right).
\end{align*}
Combining the latter two estimates we can find a constant $K>0$ such that
\begin{align*}
     \|(i)_t^n\|+ \|(ii)_t^n(b)\|
     \leq & K \Delta_n\sum_{i=1}^{\ulT}\left(1+ \frac{\|\tilde\Delta_i^n f'\|}{\sqrt{\Delta_n}}\right)\left(\frac{\|\tilde\Delta_i^n f''\|}{\sqrt{\Delta_n}}\wedge 1+\left(\frac{\|\tilde\Delta_i^n f''\|}{\sqrt{\Delta_n}}\right)^2\wedge (2v_n)^2\right).
\end{align*}
As it is
$$\left(\frac{\|\tilde\Delta_i^n f''\|}{\sqrt{\Delta_n}}\right)^2\wedge (2v_n)^2=(2 \Delta_n^{w-\frac 1 2})^2 \left(\frac{\|\tilde\Delta_i^n f''\|}{2\Delta_n^w}\wedge 1\right)^2,$$
it is now enough to employ Lemma \ref{L: Bound for truncated jump increments} to obtain for each $0\leq l\leq 1/r$ that for $f''_t$ given as in \eqref{Eq: Continuous discontinuous condition for low r} if $r<1$ and  as in \eqref{Eq: Continuous discontinuous condition for large r} if $r>1$ we have
\begin{align}\label{Est of martingale equals truncated martingale RV:IV}
    \mathbb E\left[C\frac{\|\tilde\Delta_i^n f''\|}{\Delta_n^l}\wedge 1\right]\leq K \Delta_n^{1-lr}\phi_n
\end{align}
 for a constant $K>0$, independent of $n$ and $t$ and a sequence $(\phi_n)_{n\in \mathbb N}$ that converges to $0$ as $n\to \infty$, which is independent of $t$. 
In this case, we can apply Lemma \ref{lem: Generalized BDG} and could estimate \edit{by \eqref{Est of martingale equals truncated martingale RV:IV}} with $K',K'',K'''>0$ and $q=p/(p-1)$  
 \begin{align*}
&   \mathbb E [\sup_{t\in [0,T]}  \|(i)_t^n\|+ \|(ii)_t^n(b)\| ]\\
     \leq & K\Delta_n\sum_{i=1}^{\ulT}\mathbb E\left[\left(1+ \frac{\|\tilde\Delta_i^n f'\|}{\sqrt{\Delta_n}}\right)^{p}\right]^{\frac 1{p}}\mathbb E\left[\left(\frac{\|\tilde\Delta_i^n f''\|}{\sqrt{\Delta_n}}\wedge 1\right)^q\right]^{\frac 1q}\\
     &+K\Delta_n\sum_{i=1}^{\ulT}\mathbb E\left[\left(1+ \frac{\|\tilde\Delta_i^n f'\|}{\sqrt{\Delta_n}}\right)^{p}\right]^{\frac 1{p}}\mathbb E\left[\left((2 \Delta_n^{w-\frac 1 2})^2 \left(\frac{\|\tilde\Delta_i^n f''\|}{2\Delta_n^w}\wedge 1\right)^2\right)^q\right]^{\frac 1q}\\
     \leq &K\Delta_n\sum_{i=1}^{\ulT}K'''\mathbb E\left[\frac{\|\tilde\Delta_i^n f''\|}{\sqrt{\Delta_n}}\wedge 1\right]^{\frac 1q}+ \Delta_n\sum_{i=1}^{\ul}K'''\left(2 \Delta_n^{w-\frac 1 2}\right)^2\mathbb E\left[\frac{\|\Delta_i^n f''\|}{2\Delta_n^w}\wedge 1\right]^{\frac 1q}\\
  \leq &   \edit{K \Delta_n\sum_{i=1}^{\ulT}K'''\left(K \Delta_n^{1-\frac {r}2} \phi_n\right)^{\frac 1q}+ \Delta_n\sum_{i=1}^{\ulT}K'''\left(2 \Delta_n^{w-\frac 1 2}\right)^2\left(K \Delta_n^{1-wr} \phi_n\right)^{\frac 1q}}\\
     \leq &K\left( K' T \Delta_n^{\frac {2-r}{2q}}+K'' T \Delta_n^{2w-1+\frac {1-wr}q}\right)\phi_n^{\frac 1q}
\end{align*}
with $\phi_n$ a real sequence converging to $0$ as $n\to \infty$.
Since we have that 
\edit{$(2-r)/2q>2w-1+((1-wr)/q)$} and $q=p/(p-2)$ we have shown
that by \eqref{Est of martingale equals truncated martingale RV:IV} there is a constant $K>0$ such that
\begin{align*}
   \mathbb E \left[\sup_{t\in [0,T]}  \|(i)_t^n\|+ \|(ii)_t^n(b)\|\right]
     \leq K T \Delta_n^{2w-1+\frac {1-wr}q}\phi_n^{\frac {p-1}p},
\end{align*}
which proves the claim.
\end{proof}

We continue with (ii), which follows from
\begin{lemma}\label{lem: Continuous truncated-continuous is AN}
    If Assumption \ref{In proof: Very very Weak localised integrability Assumption on the moments}(p,r) holds for $p>\max(4,2/(\frac 12-w))$ and $r\in (0,2]$,  we can find a constant $K>0$ independent of $T$ and $n$  and a sequence $(\phi_n)_{n\in \mathbb N}$ such that $\phi_n\to 0$ as $n\to\infty$ and
    \begin{align*}
          \mathbb E\left[\sup_{t\in [0,T]} \sqrt n\left\|\sum_{i=1}^{\ul} (\tilde \Delta_i^n f')^{\otimes 2}\indicator_{g_n(\Delta_i^n f')\leq u_n}-\sum_{i=1}^{\ul} (\tilde \Delta_i^n f')^{\otimes 2}\right\|_{\text{HS}}\right]
      \leq  K T \Delta_n^{\frac 12} \phi_n.
    \end{align*}
\end{lemma}
\begin{proof} By H{\"o}lder's and Markov's inequalities as well as Theorem \ref{lem: Generalized BDG} and Assumption \ref{In proof: Very very Weak localised integrability Assumption on the moments}(p,r) (observe that $p>4)$) it is
\begin{align*}
  &  \mathbb E\left[\sup_{t\in [0,T]} \sqrt n\left\|\sum_{i=1}^{\ul} (\tilde \Delta_i^n f')^{\otimes 2}\indicator_{g_n(\Delta_i^n f')\leq u_n}-\sum_{i=1}^{\ul} (\tilde \Delta_i^n f')^{\otimes 2}\right\|_{\text{HS}}\right]\\
     \leq   &  \sqrt n\sum_{i=1}^{\ulT} \Delta_n \tilde K^{\frac 12}\mathbb P\left[g_n(\Delta_i^n f')> u_n\right]^{\frac 12}.
\end{align*}
 Moreover, by Lemma \ref{lem: Generalized BDG} with $p>2/(1/2-w)$ and using \eqref{abstract g conditions} we find a constant $K>0$ such that 
\begin{align*}
  \mathbb P\left[g_n(\Delta_i^n f')> u_n\right]
  \leq   K\Delta_n^2 \Delta_n^{p(1/2-w)-2}.
\end{align*}
We obtain with $\phi_n:=\Delta_n^{p(1/2-w)-2}$ that there is a $K>0$ such that
\begin{align*}
    \mathbb E\left[ \sup_{t\in [0,T]} \sqrt n\left\|\sum_{i=1}^{\ul} (\tilde \Delta_i^n f')^{\otimes 2}\indicator_{g_n(\Delta_i^n f')\leq u_n}-\sum_{i=1}^{\ul} (\tilde \Delta_i^n f')^{\otimes 2}\right\|_{\text{HS}}\right]
     \leq &  K T\Delta_n^{\frac 12}\phi_n.
\end{align*}
\end{proof}
We can now derive (iii) from 
\begin{lemma} Let Assumption \ref{In proof: Very very Weak localised integrability Assumption on the moments}(4,r) hold.
We can find a constant $K>0$ which is independent of $T$ and $n$ such that
\begin{align*}
   \EE  \left[\sup_{t\in[0,T]}\left\|\sum_{i=1}^{\ul} (\tilde\Delta_i^n f')^{\otimes 2}-\int_0^t \Sigma_s^{\mathcal S_n}ds\right\|_{\text{HS}}\right] &\leq K T \Delta_n^{\frac 12}.
\end{align*}
\end{lemma}

\begin{proof}
We define for this proof $A^*_t=\int_0^t \mathcal S(t-s)\alpha_s' ds$ and 
$f^*_t:= f'-A^*_t$ for all $t\geq 0$.  For the proof we can assume w.l.o.g. that $f'=f^*$ due to the following argument: 
We can use H{\"o}lder's inequality to derive
\begin{align*}
 &   \mathbb E\left[\sup_{t\in [0,T]}\left\|\sum_{i=1}^{\ul} (\tilde\Delta_i^n f')^{\otimes 2}-\sum_{i=1}^{\ul} (\tilde\Delta_i^n f^*)^{\otimes 2}\right\|_{\text{HS}}\right]\\
    \edit{\leq } & \Delta_n^{-\frac 12}\sum_{i=1}^{\ulT} \mathbb E\left[\|\tilde{\Delta}_i^n A^*\|^2\right]+2\Delta_n^{-\frac 12}\left(\sum_{i=1}^{\ulT}\mathbb E\left[\|\tilde{\Delta}_i^n A^*\|^2\right]\right)^{\frac 12}\left(\sum_{i=1}^{\ulT}\mathbb E\left[\|\tilde{\Delta}_i^n f^*\|^2\right]\right)^{\frac 12}.
\end{align*}
Now using Lemma \ref{lem: Generalized BDG} we find a constant $\tilde K>0$ such that
$\mathbb E\left[\|\tilde{\Delta}_i^n f^*\|^2\right]\leq  \Delta_n \tilde K$
and by Jensen's inequality and Assumption \ref{In proof: Very very Weak localised integrability Assumption on the moments} it is
$\mathbb E\left[\|\tilde{\Delta}_i^n A^*\|^2\right]\leq \Delta_n^2 A.$
Hence, we obtain
\begin{align*}
      \mathbb E\left[\sup_{t\in [0,T]}\left\|\sum_{i=1}^{\ul} (\tilde\Delta_i^n f')^{\otimes 2}-\sum_{i=1}^{\ul} (\tilde\Delta_i^n f^*)^{\otimes 2}\right\|_{\text{HS}}\right]
   \leq &\Delta_n^{-\frac 12} \ulT A \Delta_n^2+\ulT\Delta_n^{-\frac 12} \tilde K \Delta_n^{\frac 32}.
\end{align*}

We now assume $f'=f^*$ and define
\begin{align*}
   \tilde Z_n(i) := & \Delta_n^{-\frac 12}\left((\tilde{\Delta}_i^n f^*)^{\otimes 2}-\int_{(i-1)\Delta_n}^{i\Delta_n} \Sigma_s^{\mathcal S_n} ds\right).
\end{align*}
It is simple to see under the imposed Assumption,  that $\sup_{t\in[0,T]}\Vert\sum_{i=1}^{\ul}\tilde Z_n(i)\Vert_{\text{HS}}$
has finite second moment.
Note also that $t\mapsto \psi_t=\int_{(i-1)\Delta_n}^t\mathcal S(i\Delta_n-s)\sigma_sdW_s$ is a martingale for $t\in[(i-1)\Delta_n,i\Delta_n]$. From \cite[Theorem 8.2, p.~109]{PZ2007} we deduce that 
the process $(\zeta_t)_{t\geq 0}$, with
$
    \zeta_t=  \left(\psi_t\right)^{\otimes 2}-\langle\langle \psi\rangle\rangle_t,
$
where $\langle\langle \psi\rangle\rangle_t:=\int_0^t\mathcal S(i\Delta_n-s)\Sigma_s\mathcal S(i\Delta_n-s)^* ds$,
is a martingale w.r.t. $(\mathcal{F}_t)_{t\geq 0}$, and hence,
\begin{align*}
\mathbb E\left[(\tilde{\Delta}_i^n f^*)^{\otimes 2}|\mathcal F_{(i-1)\Delta_n}\right]=&\mathbb E\left[(\psi_{i\Delta_n})^{\otimes 2}|\mathcal F_{(i-1)\Delta_n}\right]
=  \mathbb E\left[\langle\langle \psi\rangle\rangle_{i\Delta_n}|\mathcal F_{(i-1)\Delta_n}\right].
\end{align*}
So, $\mathbb E\left[\tilde Z_n(i)\right|\edit{\mathcal F_{(i-1)\Delta_n}}]=0$. 
Moreover, for $j<i$, as each $\tilde Z_n(i)$ is $\mathcal F_{(i-1)\Delta_n}$ measurable and the conditional expectation commutes with bounded linear operators, we find by using the tower property of conditional expectation that
\begin{align*}
    \mathbb E\left[\langle \tilde Z_n(i), \tilde Z_n(j)\rangle_{\text{HS}} \right]=
 \mathbb E\left[\langle \mathbb E\left[\tilde Z_n(i)|\mathcal F_{(i-1)\Delta_n}\right], \tilde Z_n^N(j)\rangle_{\text{HS}} \right]=0.
\end{align*}
Thus, we obtain 
 \begin{align*}
  \EE  \left[\left\Vert\sum_{i=1}^{\ul}\tilde Z_n(i) \right\Vert_{\text{HS}}^2 \right] \leq \sum_{i=1}^{\lfloor T/\Delta_n\rfloor} \EE\left[\Vert\tilde Z_n(i)\Vert_{\text{HS}}^2\right].
\end{align*}
Applying the  triangle and Bochner inequalities and the basic inequality $(a+b)^2\leq 2(a^2+b^2)$ we find
\begin{align*}
    \mathbb E\left[\Vert \tilde Z_n(i)\Vert_{\text{HS}}^2\right]
    &\leq 4 \int_{(i-1)\Delta_n}^{i\Delta_n}\mathbb E\left[\Vert \mathcal S (i\Delta_n-s)\sigma_s \Vert_{\text{HS}}^4\right]ds. 
\end{align*}
Summing up, we obtain
\begin{align*}
 \EE  \left[\sup_{t\in[0,T]}\left\Vert\sum_{i=1}^{\ul}\tilde Z_n(i) \right\Vert_{\text{HS}}^2 \right] 
   \leq 4 \sup_{n\in\mathbb N}\int_0^T\mathbb E\left[\Vert \mathcal S (i\Delta_n-s)\sigma_s \Vert_{\text{HS}}^4\right]ds\leq 4\sup_{r\leq \Delta_n}\|\mathcal S(r)\|_{\text{op}}AT.
\end{align*}

\end{proof}

It remains to prove (iv), which follows from
\begin{lemma}\label{L: Temporal Discretisation error induced by the semigroup}
There is a constant $K>0$ such that
\begin{align*}
   & \sup_{t\in[0,T]}\left\|\int_0^t \Sigma_s^{\mathcal S_n} ds-[X^C,X^C]_t ds\right\|_{\text{HS}}
    \leq  K\int_{0}^{T}\sup_{r\in [0, \Delta_n]}\frac{\|(\mathcal S(r)-I)\Sigma_{s}\|_{\text{HS}}}{r^{\gamma}} ds.
\end{align*}
In particular, if Assumption \ref{As: spatial regularity}($\gamma$) holds for $\gamma \in (0,\frac 12]$, then
 we have
  $$\sup_{t\in[0,T]}\left\|\int_0^t \Sigma_s^{\mathcal S_n} ds-[X^C,X^C]_t\right\|_{\text{HS}}=\mathcal O_p(\Delta_n^{\gamma})$$
  and under Assumption \ref{In proof II}($\gamma$) there is a constant $K>0$ which is independent \edit{of} $T$ and $n$ such that
\begin{equation}\label{Eq: Integratedvol temporal discretization}
      \mathbb E \left[\sup_{t\in[0,T]}\left\|\int_0^t \Sigma_s^{\mathcal S_n} ds-[X^C,X^C]_t\right\|_{\text{HS}}\right]\leq K T\Delta_n^{\gamma} .
  \end{equation}
\end{lemma}
\begin{proof}  It is $[X^C,X^C]_t= \int_0^t  \Sigma_s  ds$ and, hence, 
    \begin{align*}
  n^{\gamma} \left\| \Sigma_s^{\mathcal S_n} -\Sigma_s\right\|_{\text{HS}} 
      \leq  &(\sup_{s\geq 0}\|\mathcal S(s)\|_{\text{op}} +1)\sup_{r\in [0, \Delta_n]}\frac{\|(\mathcal S(r)-I)\Sigma_{s}\|_{\text{HS}}}{r^{\gamma}}. 
    \end{align*}
   This proves the claim under Assumption \ref{As: spatial regularity}($\gamma$)(i). Under Assumption \ref{As: spatial regularity}($\gamma$)(ii), we obtain 
    \begin{align*}
  n^{\gamma} \mathbb E[\left\| \Sigma_s^{\mathcal S_n} -\Sigma_s\right\|_{\text{HS}} ]
      \leq  &(\sup_{s\geq 0}\|\mathcal S(s)\|_{\text{op}} +1)\sup_{r\in [0, \Delta_n]}\frac{\mathbb E[\|(\mathcal S(r)-I)\Sigma_{s}\|_{\text{HS}}]}{r^{\gamma}} .
    \end{align*}
    The bound \eqref{Eq: Integratedvol temporal discretization} under Assumption \ref{In proof II}($\gamma$) follows immediately.
\end{proof}

\subsection{Proof of \ref{T: CLT without disc}}
This section contains the proof of Theorem \ref{T: CLT without disc}.
  \begin{proof}[Proof of \edit{Theorem} \ref{T: CLT without disc}]
We first prove that Condition \eqref{Abstract CLT assumption}(i) implies Assumption \ref{As: spatial regularity}(1/2)(i). For that, we can derive from the basic inequality $\|AB^*\|_{L_{\text{HS}}(H)}\leq \|A\|_{\text{HS}}\|B\|_{\text{op}},$
whenever $A$ is a Hilbert-Schmidt and $B
$ is a bounded linear operator,
that
\begin{align*}
\int_0^T \sup_{r> 0}\frac{\|(I-\mathcal S(r))\Sigma_s\|_{\text{HS}}}{r}ds  
\leq & \left(\int_0^T \sup_{r> 0}\frac{\|(I-\mathcal S(r))\sigma_s\|_{\text{op}}^2}{r^2}ds   \right)^{\frac 12}\left(\int_0^T\|\sigma_s\|_{\text{HS}}^2ds\right)^{\frac 12}.
\end{align*}
The latter is almost surely finite by Condition \eqref{Abstract CLT assumption} and since $\sigma$ is stochastically integrable. 
  
 Deriving Assumption \ref{As: spatial regularity}(1/2)(i) from Condition \eqref{Abstract CLT assumption}(i) works similarly, so we skip it.
  
Let us now prove the central limit theorem.  We have proven under Assumption \ref{In proof: Very very Weak localised integrability Assumption on the moments}(p,r) for $p>max(2/(1-2w),(2-2wr/(2(2-r)w+1))$ with $r<1$ and $w\in[1/(4-2r),1/2) $  by  \edit{(i)} and \edit{(ii)} that
  $$\sup_{t\in [0,T]} \left\|SARCV_t^n(u_n,-)-\sum_{i=1}^{\ul} (\tilde \Delta_i^n f')^{\otimes 2}\right\|_{\text{HS}}= o_p(\Delta_n^{\frac 12}).$$
  By virtue of Lemma \ref{L: Localisation}(iv), we obtain that this must also hold if just Assumption \ref{As: H} is valid instead of Assumption \ref{In proof: Very very Weak localised integrability Assumption on the moments}(p,r). 
  Therefore,  we have
  \begin{align*}
      \sqrt n\left(SARCV_t^n(u_n,-)- [X^C,X^C]_t\right)
      = & \sqrt n\left(\sum_{i=1}^{\ul} (\tilde \Delta_i^n f')^{\otimes 2}-\int_0^t \Sigma_s  ds\right) + o_p(1).
  \end{align*}
 Now the claim follows by Theorem 3.5 in \cite{BSV2022}, which implies 
 \begin{align*}
    \sqrt n\left(\sum_{i=1}^{\ul} (\tilde \Delta_i^n f')^{\otimes 2}-\int_0^t \Sigma_s  ds\right)\overset{st.}{\longrightarrow} \mathcal N(0,\mathfrak Q_t).
 \end{align*}
  \end{proof}

\subsection{Proof of \ref{T: Long-time no disc}}
This section contains the proof for Theorem \ref{T: Long-time no disc}.
\begin{proof}[Proof of \ref{T: Long-time no disc}]
The proof for the long time asymptotics is an immediate corollary of Theorems \ref{T: Uniform convergence theorem for realized variation}. To see that, observe that
\begin{align*}
    \frac 1T SARCV_T^n(u_n,-)-\mathcal C
    = & \left(\frac 1T SARCV_T^n(u_n,-)-\frac 1T \int_0^T \Sigma_s  ds\right)+\left(\frac 1T \int_0^T \Sigma_s ds-\mathcal C\right).
\end{align*} 
The first term is $\mathcal O_p(\Delta_n^{\gamma})$ by Theorem \ref{T: Uniform convergence theorem for realized variation},  the second term is $\mathcal O_p(a_T)$ by Assumption,  hence,  the proof.
\end{proof}

\subsection{Proof of Corollary \ref{C: Results for realized variation}}
This section contains the proof for Corollary \ref{C: Results for realized variation}.
 \begin{proof}[Proof of Corollary \ref{C: Results for realized variation}]
     We just have to show that we have
     $f_t=f_0+\int_0^t \mathcal A f_s ds +X_t
     $
     since in that case Theorems \ref{T: Abstract identification of quadr var}, \ref{T: Identification of continuous and disc quadr var}, \ref{T: rates of convergence},\ref{T: CLT without disc} and \ref{T: Long-time no disc}
hold with $\mathcal S=I$ the identity.
 Due to the moment assumption on the volatilties, we can apply the stochastic Fubini theorems and obtain that for each $h\in D(\mathcal A^*)$ it is $  \langle f_t,h\rangle= \langle X_t,h\rangle +\int_0^t \langle \mathcal A f_u, h\rangle du$.
This proves the claim.
 \end{proof}

\section{Auxiliary Technical Results}\label{Sec: Auxiliary Technical Results}
In this section we provide proofs for several technical results that were used in the proofs of Theorems \ref{T: Abstract identification of quadr var},  \ref{T: Identification of continuous and disc quadr var}, \ref{T: rates of convergence}, \ref{T: CLT without disc},  \ref{T: Long-time no disc}, and \ref{C: Results for realized variation}.. We assume that $f$ and $X$ are defined as in Section \ref{mild Ito process}. 
We will start with a Burkholder-Davis-Gundy-type estimate, which is used in the proof of Theorem \ref{T: Uniform convergence theorem for realized variation} but might be of interest in its own.
\begin{lemma}\label{lem: Generalized BDG}
    Let Assumption \ref{In proof: Very very Weak localised integrability Assumption on the moments}(p,r) hold for $r=2$ and $p>1$. Then, if $\gamma_s(z)\equiv 0$ or if $\gamma_s(z)\neq 0$ for $p=2$ we can find a constant $K>0$ which is independent of $i$ and $n$ such that
    $$\mathbb E\left[\left\|\tilde{\Delta}_i^n f\right\|^p\right]\leq \tilde K \Delta_n^{\frac p2}.$$

\end{lemma}

\begin{proof}
    For  $t\geq (i-1)\Delta_n$, define the square-integrable martingale
    $$M_t^{\mathcal S}:= \int_{(i-1)\Delta_n}^{t} \mathcal S(t-s)\sigma_s dW_s+\int_{(i-1)\Delta_n}^{t} \int_{H\setminus \{ 0\}} \mathcal S(t-s) \gamma_s(z)(N-\nu)(dz,ds)
    $$
    observe that
    \begin{align*}
        \mathbb E\left[\left\|\tilde{\Delta}_i^n f\|^p\right\|\right]^{\frac 1p}
        \leq &  \left(\mathbb E\left[\left\|\int_{(i-1)\Delta_n}^{i\Delta_n} \mathcal S(t-s)\alpha_s ds\right\|^p\right]^{\frac 1p}\right)+\left(\mathbb E\left[\left\|M_{i\Delta_n}^{\mathcal S}\right\|^p\right]^{\frac 1p}\right).
    \end{align*}
    For the first summand we find
    \begin{align*}
        \mathbb E\left[\left\|\int_{(i-1)\Delta_n}^{i\Delta_n} \mathcal S(t-s)\alpha_s ds\right\|^p\right]
        \leq & \Delta_n^{p}  A^p \sup_{r\leq \Delta_n}\|\mathcal S(r)\|_{\text{op}}^p.
    \end{align*}
    For the second summand, we can use the Burkholder-Davis-Gundy inequality c.f. (Theorem 1 in \cite{Marinelli2016}) in the case that $p\neq 2$ and $\gamma\equiv 0$ to obtain a constant $C_p>1$, just depending on $p$ such that
    \begin{align*}
      \mathbb E\left[\left\|\int_{(i-1)\Delta_n}^{i\Delta_n} \mathcal S(t-s)\sigma_s dW_s\right\|^p\right]
        \leq & C_p \Delta_n^{\frac p2} A^p\sup_{r\leq \Delta_n}\|\mathcal S(r)\|_{\text{op}}^p
    \end{align*}
    and in the case $\gamma\neq 0$ and $p=2$
    \begin{align*}
        \mathbb E\left[\left\|\int_{(i-1)\Delta_n}^{i\Delta_n} \int_{H\setminus \{ 0\}} \mathcal S(t-s) \gamma_s(z)(N-\nu)(dz,ds)\right\|^2\right]
        \leq  \Delta_n A^2\sup_{r\leq \Delta_n}\|\mathcal S(r)\|_{\text{op}}^2.
    \end{align*}
    Summing up, we obtain in both cases a constant $\tilde C_p$ such that
    $$ \mathbb E\left[\left\|\tilde{\Delta}_i^n f\|^p\right\|\right]\leq  \tilde C_p\Delta_n^{\frac p2} A^p.$$

\end{proof}

\begin{lemma}\label{L: Projection convergese uniformly on the range of volatility}
 Suppose that Assumption \ref{As: localised integrals} holds.
 Let $(e_j)_{j\in\mathbb N}$ be an orthonormal basis of $H$ and $P_N$ be the projection onto $ span\{e_{j}:j=1,...,N\}$.
 Then we have for all $T>0$
\begin{align}
\lim_{N\to\infty}&\int_0^T\sup_{r\in [0,T]}\mathbb E\left[\|(I-P_N)\mathcal S(r)\alpha_s\|\right]ds= 0,\label{drift-vanishes asymptotically}\\
\lim_{N\to\infty}&\int_0^T\sup_{r\in [0,T]}\mathbb E\left[\|(I-P_N)\mathcal S(r)\sigma_s\|_{\text{HS}}^2\right]ds= 0,\label{Volatiltity-vanishes asymptotically}\\
\lim_{N\to\infty}&\int_0^T\sup_{r\in [0,T]}\int_{H\setminus\{ 0\}}\mathbb E\left[\|(I-P_N)\mathcal S(r)\gamma_s(z)\|^2\right] F(dz)ds= 0.\label{jumps-vanish asymptotically}
\end{align}
In particular, it is
\begin{equation}\label{increment-vanishes asymptotically}
    \lim_{N\to\infty} \sup_{n\in \mathbb N}\sum_{i=1}^{\ulT}\mathbb E\left[\|(I-P_N)\tilde{\Delta}_i^n f\|^2\right]=0.
\end{equation}
\end{lemma}

\begin{proof}
Observe that, for $s,r\in[0,T]$ fixed, we have almost surely as $N\to\infty$ that
\begin{align*}
   \|(I-P_N)\mathcal S(r)\sigma_{s}\|_{\text{HS}}^2=\|\sigma_{s}^*\mathcal S(r)^*(I-P_N)\|_{\text{HS}}^2= \sum_{k=N+1}^{\infty} \|\sigma_{s}^*\mathcal S(r)^*e_k\|^2\to 0
\end{align*}
since $\sigma_{t}^*\mathcal S(s)^*$ is almost surely a Hilbert-Schmidt operator. As $I-P_N\to 0$ strongly as $N\to\infty$ we also have for all $z\in H\setminus\{0\}$ and $s,r\geq 0$ that $
   \|(I-P_N)\mathcal S(r)\gamma_{s}(z)\|\to 0$ and $
   \|(I-P_N)\mathcal S(r)\alpha_{s}\|\to 0$.
Moreover, the functions $\mathfrak H_s(r):= \|(I-P_N)\mathcal S(r)\alpha_s\|$, $\mathfrak F_s(r)= \|(I-P_N)\mathcal S(r)\sigma_{s}\|_{\text{HS}}$ 
and $\mathfrak G_s(r)= \|(I-P_N)\mathcal S(r)\gamma_{s}(z)\|$  are continuous in $r$. For $\mathfrak F_s$ this follows from
\begin{align*}
    |\mathfrak F_s(r_1)-\mathfrak F_s(r_2)|
    \leq  \sup_{r\in[0,T]}\|\mathcal S(r)\|_{\text{op}} \sup_{r\leq|r_1-r_2|}\left\|\left(I-\mathcal S(r)\right)\sigma_s\right\|_{\text{HS}}.
\end{align*}
The latter converges to $0$, as $r_1\to r_2$, by Proposition 5.1 in \cite{Benth2022}. For $\mathfrak G_s$ we have
\begin{align*}
    |\mathfrak G_s(r_1)-\mathfrak G_s(r_2)| 
    \leq  \sup_{r\in[0,T]}\|\mathcal S(r)\|_{\text{op}} \sup_{r\leq|r_1-r_2|}\left\|\left(I-\mathcal S(r)\right)\gamma_s(z)\right\|.
\end{align*}
Analogously we obtain
\begin{align*}
    |\mathfrak H_s(r_1)-\mathfrak H_s(r_2)|
    \leq \sup_{r\in[0,T]}\|\mathcal S(r)\|_{\text{op}} \sup_{r\leq|r_1-r_2|}\left\|\left(I-\mathcal S(r)\right)\alpha_s\right\|.
\end{align*}
Both converge to $0$, as $r_1\to r_2$, due to the strong continuity of the semigroup.
As we also have
$$\|(I-P_N)\mathcal S(s)\sigma_{t}\|_{\text{HS}}\geq  \| (I-P_{N+1})(I-P_N)\mathcal S(r)\sigma_s\|_{\text{HS}}=\|(I-P_{N+1})\mathcal S(r)\sigma_s\|_{\text{HS}},$$
and
$$\|(I-P_N)\mathcal S(r)\gamma_s(z)\|\geq  \| (I-P_{N+1})(I-P_N)\mathcal S(r)\gamma_s(z)\|=\|(I-P_{N+1})\mathcal S(r)\gamma_s(z)\|,$$
as well as
$$\|(I-P_N)\mathcal S(r)\alpha_s\|\geq  \| (I-P_{N+1})(I-P_N)\mathcal S(r)\alpha_s\|=\|(I-P_{N+1})\mathcal S(r)\alpha_s\|,$$
we find by virtue of Dini's theorem (cf. Theorem 7.13 in \cite{Rudin1976}) that almost surely
$$\lim_{N\to\infty}\sup_{r\in[0,T]}\left\|(I-P_N)\mathcal S(r)\sigma_s\right\|_{\text{HS}}+\left\|(I-P_N)\mathcal S(r)\gamma_s(z)\right\|+\left\|(I-P_N)\mathcal S(r)\alpha_s\right\|=0.$$
Now \eqref{drift-vanishes asymptotically}, \eqref{Volatiltity-vanishes asymptotically} and \eqref{jumps-vanish asymptotically} follow from dominated convergence.

Let us now prove \eqref{increment-vanishes asymptotically}. For that observe that by \edit{the triangle inequality, the basic inequality $(a+b+c)^2\leq 3(a^2+b^2+c^2)$ as well as }It{\^o}'s isometry  we find
\begin{align*}
  &  \sum_{i=1}^{\ulT}\mathbb E\left[\|(I-P_N)\tilde{\Delta}_i^n f\|^2\right]\\
   \leq & 3\mathbb E\left[\int_0^TA\|(I-P_N)\alpha_s\| +  \|(I-P_N)\sigma_s\|_{\text{HS}}^2 + \int_{H\setminus \{ 0\}}\|(I-P_N)\gamma_s(z)\|^2 F(dz)ds\right].
\end{align*}
Using \eqref{drift-vanishes asymptotically}, \eqref{Volatiltity-vanishes asymptotically} and \eqref{jumps-vanish asymptotically} we obtain \eqref{increment-vanishes asymptotically}.
\end{proof}

To estimate increments corresponding to the jump part of our process  we need the subsequent Lemma.  
\begin{lemma}\label{L: Bound for truncated jump increments}
    Let Assumption \ref{In proof: Very very Weak localised integrability Assumption on the moments}(p, r) hold for some $r\in (0,2]$ and a $p\geq 2$. If $r\leq 1$ denote $f''_t:= \int_0^t \int_{H\setminus\{ 0\}} \mathcal S(t-s)\gamma_s(z) N(dz,ds)$ and if $r>1$ $f''_t:= \int_0^t \int_{H\setminus\{ 0\}} \mathcal S(t-s)\gamma_s(z) (N-\nu)(dz,ds)$. Then we have 
    for any constant $K>0$, $0<l<1/r$ that there is a real-valued sequence $(\phi_n)_{n\in \mathbb N}$ converging to $0$ as $n\to \infty$ such that
    \begin{align}
        \mathbb E\left[\left(K\frac{\|\tilde{\Delta}_i^n f''\|}{\Delta_n^{l}}\right)\wedge 1\right]\leq \Delta_n^{1-lr}\phi_n.
    \end{align}
\end{lemma}

\begin{proof}
We start with the case that $r<1$ and essentially follow the steps in the proof of Lemma 2.18 in \cite{JacodProtter2012} with slight adjustments for our setting. Fix $i\in \{ 1,...,\ul\}$. Define
$$\gamma:=\int_{\Gamma(z)\leq 1} \Gamma(z) F(dz)+F(\{z:\Gamma(z)>1\}.$$
This is finite due to Assumption \ref{In proof: Very very Weak localised integrability Assumption on the moments}(p,r), since we have
$$\int_{\Gamma(z)\leq 1} \Gamma(z) F(dz)=\mathbb E\left[\int_{\Gamma(z)\leq 1} \|\gamma_s(z)\|^r F(dz)\right]<\infty$$
and
\begin{align*}
    F(\{z:\Gamma_s(z)>1\} \leq & \mathbb E\left[\int_{\Gamma_s(z)> 1} \|\gamma_s(z)\|^r F(dz)\right] <\infty.
\end{align*}
Now set
for $\epsilon > 0$
$$A(\epsilon):= \int_{(i-1)\Delta_n}^{i\Delta_n} \int_{H\setminus \{ 0\}}\mathcal S(i\Delta_n-s)\gamma_s(z) \indicator_{\Gamma(z)\leq \epsilon} N(dz,ds)$$
and 
$$N(\epsilon):= \int_{(i-1)\Delta_n}^{i\Delta_n} \int_{H\setminus \{ 0\}}\indicator_{\Gamma(z)>\epsilon} N(dz,ds)=N\left(\{z\in H\setminus \{0\}:\Gamma(z)>\epsilon\},[(i-1)\Delta_n,i\Delta_n]\right).$$
Since this is integer-valued, we then have, since 
$$K\frac{\|\tilde\Delta_i^n f''\|}{\Delta_n^l}\wedge 1\leq \left(K\frac{\|\tilde\Delta_i^n f''\|}{\Delta_n^l}\wedge 1\right)^r\leq\indicator_{N(\epsilon)>0}+ K^r\Delta_n^{-rl} \|A(\epsilon)\|^r$$
that
\begin{align*}
    \mathbb E\left[K\frac{\|\tilde\Delta_i^n f''\|}{\Delta_n^l}\wedge 1\right] = & \mathbb P\left[N(\epsilon)>0\right]+C^r\Delta_n^{-rl} \mathbb E\left[\|A(\epsilon)\|^r\right].
\end{align*}
Observe that $N(\epsilon)\geq 0$ is an integer valued random variable and for $0<\epsilon<1$
$$\indicator_{\Gamma(z)>\epsilon}\leq \epsilon^{-1}(\Gamma(z)\indicator_{\Gamma(z)\leq 1}
+\indicator_{\Gamma(z)> 1})
,$$ so for $0<\epsilon<1$ we have
\begin{align*}
\mathbb P[  N_{i\Delta_n}(\epsilon)>0]
\leq   \Delta_n \epsilon^{-1} 
\gamma.
\end{align*}
For the second summand we have 
\begin{align*}
    \mathbb E\left[\|A(\epsilon)\|^r\right] 
    \leq & \Delta_n  \left(\sup_{u\leq \Delta_n}\|\mathcal S(u)\|_{\text{op}}\int_{H\setminus \{ 0\}}\Gamma(z)^r\indicator_{\Gamma(z)\leq \epsilon} F(dz)\right).
\end{align*}
Summing up, we obtain
\begin{align*}
    \mathbb E\left[K\frac{\|\Delta_i^n f''\|}{\Delta_n^l}\wedge 1\right]
    \leq & \Delta_n \epsilon^{-1} \gamma+ K^r \Delta_n^{1-rl}\left(\sup_{u\leq \Delta_n}\|\mathcal S(u)\|_{\text{op}}\int_{H\setminus \{ 0\}}\Gamma(z)^r\indicator_{\Gamma(z)\leq \epsilon} F(dz)\right).
\end{align*}
Choosing $\epsilon:=\Delta_n^{rl'}$ for $0<l'<l$, this yields the claim.

Now we come to the case $r\geq1$. For $\epsilon>0$ 
we define $N(\epsilon)$ as before and introduce 
$$B(\epsilon):=-\int_{(i-1)\Delta_n}^{i\Delta_n}\int_{\Gamma(z)> \epsilon} \mathcal S(i\Delta_n-s) \gamma_s(z) \nu(dz,ds)$$
as well as 
$$M(\epsilon):=\int_{(i-1)\Delta_n}^{i\Delta_n}\int_{\Gamma(z)\leq \epsilon} \mathcal S(i\Delta_n-s) \gamma_s(z) (N-\nu)(dz,ds).$$
For some constant $K>0$ it is then
$$\mathbb E\left[\frac{\|\tilde \Delta_i^n f''\|}{\Delta_n^l}\wedge 1\right]\leq \mathbb E\left[\left(\frac{\tilde \Delta_i^n f''}{\Delta_n^l}\wedge 1\right)^r\right]\leq K\left(\indicator_{N(\epsilon)>0}+\Delta_n^{-rl}\mathbb E[M(\epsilon)^r]+\Delta_n^{-rl}\mathbb E[B(\epsilon)^r]\right).$$
By the Burkholder-Davis-Gundy inequality \cite[Thm.1]{Marinelli2016}  applied to the martingale
$$M_u= \int_{(i-1)\Delta_n}^u \int_{\Gamma(z)\leq \epsilon} \mathcal S(i\Delta_n-s) \gamma_s(z)\indicator_{[(i-1)\Delta_n,\infty)} (N-\nu)(dz,ds)$$ on $u\in [0,i\Delta_n]$ 
we obtain $K_1>0$ such that
\begin{align*}
    \mathbb E\left[\|M(\epsilon)\|^r\right] 
        \leq  K_1 \mathbb E\left[\int_{(i-1)\Delta_n}^{i\Delta_n} \int_{\Gamma(z)\leq \epsilon} \|\gamma_s(z)\|^r N(dz,ds)\right]
        \leq  K_1 \Delta_n \int_{\Gamma(z)\leq \epsilon} \Gamma(z) F(dz).
    \end{align*}
       Moreover, we find a constant $K_2>0$ such that
   $$\mathbb E\left[\|B(\epsilon)\|^r\right]= \mathbb E\left[\left(\int_{(i-1)\Delta_n}^{i\Delta_n}\int_{\Gamma(z)>\epsilon}\|\gamma_s(z) \|F(dz)ds\right)^r\right]\leq K_2\int_{\Gamma(z)> \epsilon} \Gamma(z)^rF(dz) \Delta_n.$$
   Summing up and taking into account the previous bound on $N(\epsilon)$ we get a $\tilde K>0$ such that
   \begin{align*}
    \mathbb E\left[K\frac{\|\Delta_i^n f''\|}{\Delta_n^l}\wedge 1\right]
    \leq & \tilde K\left( \Delta_n \epsilon^{-1}+\Delta_n^{1-rl}\int_{H\setminus \{0\}} \Gamma(z)^r F(dz)\right).
\end{align*}
Choosing $\epsilon:=\Delta_n^{rl'}$ for $0<l'<l$, this yields the claim.
\end{proof}

To derive convergence of the finite-dimensional distributions in Theorems \ref{T: Abstract identification of quadr var} and \ref{T: Identification of continuous and disc quadr var} we used
\begin{lemma}\label{L: Reduction to semimartingales is possible}
Suppose that Assumption \ref{As: localised integrals} holds.
 Then for $h,g\in D(\mathcal A^*)$ we have with $\Delta_i^n X:=X_{i\Delta_n}-X_{(i-1)\Delta_n}$ that
\begin{align}\label{Finite Projections of Power variations are essentially functions of differences of Martingales 0}
\sum_{i=1}^{\ul} \langle \tilde{\Delta}_i^n f^{\otimes 2},h\otimes g\rangle_{L_{\text{HS}}(H)}=\sum_{i=1}^{\ul}\langle \Delta_i^n X^{\otimes 2},h\otimes g\rangle_{\text{HS}}+ R_n^t
\end{align}
where $R_n^t$ depends on $h$ and $g$ and for all $0<l<1$ it is $\Delta_n^{-l}R_n^t\overset{u.c.p.}{\longrightarrow}0.$ as $n\to\infty$.

Moreover, under Assumption \ref{As: SH}(2) we have that
\begin{align}\label{Finite Projections of Power variations are essentially functions of differences of Martingales}
\left|\sum_{i=1}^{\ul} \langle \tilde{\Delta}_i^n f^{\otimes 2},h\otimes g\rangle_{\text{HS}}\indicator_{g_n(\tilde{\Delta}_i^n f)\leq  u_n}
-\sum_{i=1}^{\ul}\langle \Delta_i^n X^{\otimes 2},h\otimes g\rangle_{\text{HS}}\indicator_{\|\Delta_i^n X\|\leq  u_n}\right|\overset{u.c.p.}{\longrightarrow}0.
\end{align}
\end{lemma}

\begin{proof}
Let $g\in D(\mathcal A^*)$.
Define 
\begin{align}
    \Delta_n^t a^g&:=\int_t^{t+\Delta_n}\int_t^u\frac d{du} \langle \alpha_s,\mathcal S(u-s)^*g\rangle dsdu+ \int_t^{t+\Delta_n} \int_t^u  \langle \sigma_s, \mathcal  A^*\mathcal S(u-s)^* g\rangle dW_sdu\notag\\
    &\qquad +\int_t^{t+\Delta_n}\int_t^u \int_{H\setminus \{0\}} \langle \gamma_s(z), \mathcal A^*\mathcal S(u-s)^*g\rangle (N-\nu)(dz,ds)du.\notag
\end{align}
Using the stochastic Fubini theorem we obtain
\begin{align*}
    \langle f_{t+\Delta_n}-\mathcal S(\Delta_n)f_t,g\rangle
       = & \Delta_n^t a^g +\langle X_{t+\Delta_n}-X_t,g\rangle.
\end{align*}
Hence, for $h,g\in D(\mathcal A^*)$
\begin{align}\label{Eq: Y increment decomp wrt to S}
 &   \langle \tilde{\Delta}_{i}^n f ,g\rangle \langle \tilde{\Delta}_{i}^n f ,h\rangle\notag\\
     =  &  \langle \Delta_{i}^n X ,g\rangle \langle \Delta_{i}^n X ,h\rangle+\langle \Delta_{i}^n X ,g\rangle \Delta_n^{(i-1)\Delta_n} a^{h}
     +\langle \Delta_{i}^n X ,h\rangle \Delta_n^{(i-1)\Delta_n} a^{g}+\Delta_n^{(i-1)\Delta_n} a^{h}\Delta_n^{(i-1)\Delta_n} a^{g}\notag\\
     =:&(1)_i^n+(2)_i^n+(3)_i^n+(4)_i^n.
\end{align}
To prove the claim we need to prove
\begin{itemize}
    \item[($\star$)]Under Assumption \ref{As: localised integrals}, we have that $\Delta_n^{-1}\mathbb E\left[\left|\sum_{i=1}^{\ul}(2)_i^n+(3)_i^n+(4)_i^n\right|\right]$ is a bounded sequence. This proves \eqref{Finite Projections of Power variations are essentially functions of differences of Martingales 0}.
\end{itemize}
For that, let Assumption \ref{As: localised integrals} be valid.
We have
\begin{align}\label{technical estimate for projected increments I}
&\mathbb E\left[|\langle \Delta_{i}^n X ,g\rangle \Delta_n^{(i-1)\Delta_n} a^{h}+\langle \Delta_{i}^n X ,h\rangle \Delta_n^{(i-1)\Delta_n} a^{g}+\Delta_n^{(i-1)\Delta_n} a^{g}\Delta_n^{(i-1)\Delta_n} a^{h}|\right]\\
\leq &  (\|g\|+\|h\|) \mathbb E\left[\| \Delta_{i}^n X \|^2\right]^{\frac 12}\left(\mathbb E\left[\left( \Delta_n^{(i-1)\Delta_n}a^{g}\right)^2\right]^{\frac 12}+\mathbb E\left[\left( \Delta_n^{(i-1)\Delta_n}a^{h}\right)^2\right]^{\frac 12}\right)\notag\\
& \qquad+ \left(\mathbb E\left[\left( \Delta_n^{(i-1)\Delta_n}a^{h}\right)^2\right]\mathbb E\left[\left( \Delta_n^{(i-1)\Delta_n}a^{g}\right)^2\right]\right)^{\frac 12}.\notag
\end{align}
Since  by Assumption \ref{As: localised integrals} $$\mathbb E\left[\left\|\int_{(i-1)\Delta_n}^{i\Delta_n} \alpha_s ds\right\|^2\right]\leq A \int_{(i-1)\Delta_n}^{i\Delta_n} \mathbb E\left[\left\|\alpha_s\right\|\right] ds$$
and by \edit{the triangle inequality, the basic inequality $(a+b+c)^2\leq 3(a^2+b^2+c^2)$ as well as } It{\^o}'s isometry we obtain
\begin{align*}
     \mathbb E\left[\| \Delta_{i}^n X \|^2\right]
     \leq &  3\int_{(i-1)\Delta_n}^{i\Delta_n} \mathbb E\left[A\left\|\alpha_s \right\|+\left\|\sigma_s \right\|_{\text{HS}}^2+\int_{H\setminus\{ 0\}}\left\|\gamma_s(z)\right\|^2 F(dz)\right]ds.
\end{align*}
Moreover,
\begin{align*}
 &   \mathbb E\left[\left(\Delta_n^{(i-1)\Delta_n} a^{g}\right)^2\right]\\
    \leq & 3\Delta_n^2 \|\mathcal A^*g\|^2\sup_{t\in [0,\Delta_n]}\|\mathcal S(t)\|_{\text{op}}^2  \int_{(i-1)\Delta_n}^{i\Delta_n} \left(A\mathbb E\left[\|\alpha_s\|\right] + \mathbb E\left[\|\sigma_s\|_{\text{HS}}^2 \right]+ \int_{H\setminus \{0\}}  \mathbb E[\|\gamma_s(z)\|^2]  F(dz)\right)ds.
\end{align*}
Combining the latter two estimates with \eqref{technical estimate for projected increments I} we find a constant $K>0$ depending on $g$ and $h$ such that
\begin{align}\label{Eq: In proof of reduction to semimartingales is possible}
\sum_{i=1}^{\ul}\mathbb E\left[\left|\langle \Delta_{i}^n X ,g\rangle \Delta_n^{(i-1)\Delta_n} a^{h}+\langle \Delta_{i}^n X ,h\rangle \Delta_n^{(i-1)\Delta_n} a^{g}+\Delta_n^{(i-1)\Delta_n} a^{g}\Delta_n^{(i-1)\Delta_n} a^{h}\right|\right]
\leq  K A^2 \Delta_n.
\end{align}
This proves ($\star$) and, hence, \eqref{Finite Projections of Power variations are essentially functions of differences of Martingales 0}.

We will now prove \eqref{Finite Projections of Power variations are essentially functions of differences of Martingales}.
We can then decompose
\begin{align*}
&\left|\sum_{i=1}^{\ul} \langle \tilde{\Delta}_i^n f^{\otimes 2},h\otimes g\rangle_{L_{\text{HS}}(H)}\indicator_{g_n(\tilde{\Delta}_i^n f)\leq  u_n}
-\sum_{i=1}^{\ul}\langle \Delta_i^n X^{\otimes 2},h\otimes g\rangle_{L_{\text{HS}}(H)}\indicator_{\|\Delta_i^n X\|\leq  u_n}\right|\\
\leq &2 \sum_{i=1}^{\ul} \left|\langle \tilde{\Delta}_i^n f^{\otimes 2}-\Delta_i^n X^{\otimes 2} ,h\otimes g\rangle_{\text{HS}}\right|+\sum_{i=1}^{\ul}\|\tilde \Delta_i^n f\|^2\indicator_{c\|\tilde \Delta_i^n f\|\leq  u_n<\|\Delta_i^n X\|}\\
&\qquad+\sum_{i=1}^{\ul}\| \tilde\Delta_i^n X\|^2\indicator_{\|\Delta_i^n X\|\leq  u_n<C\|\tilde \Delta_i^n f\|}\\
=: &(i)_t^n+(ii)_t^n+(iii)_t^n.
\end{align*}
\edit{For the first term, using \eqref{Eq: Y increment decomp wrt to S} we find that 
$$\mathbb E[(i)_t^n]= 2\sum_{i=1}^{\ul}\mathbb E\left[\left|(2)_i^n+(3)_i^n+(4)_i^n\right|\right]$$
which converges to $0$ u.c.p. as $n\to \infty $ using the bound in \eqref{Eq: In proof of reduction to semimartingales is possible}.}
We prove convergence of the third summand (convergence of second is analogous).
For that, we assume without loss of generality that $g_n(h)\leq \|h\|$ for all $h\in H$ and $n\in \mathbb N$, so $C=1$ and
we introduce the notation
$X''_t=\int_0^t \int_{H\setminus\{0\}} \gamma_s(z) (N-\nu)(dz,ds) $ and $X'_t=X_t-X''_t$, such that $X'$ is a continuous process in $H$. Observe that $\|\Delta_i^n X'\|\leq u_n$ together with $\|\Delta_i^n X\|\leq u_n$ implies $\|\Delta_i^n X''\|\leq 2u_n$ and we obtain
\begin{align*}
   \sum_{i=1}^{\ulT} (iii)_t^n
\leq & 2 \sum_{i=1}^{\ulT}  \| \Delta_i^n X'\|^2\indicator_{  u_n<\|\tilde \Delta_i^n f\|} + \sum_{i=1}^{\ulT} 4u_n^2 \left(\frac{\| \Delta_i^n X''\|^2}{4u_n^2}\wedge 1\right)
=:  (iii)^n(a)+(iii)^n(b).
\end{align*}
For the first summand, we can make use of Markov's inequality and Assumption \ref{As: SH} and use Lemma \ref{lem: Generalized BDG} to find a $\tilde K>0 $  such that
\begin{align*}
    \sum_{i=1}^{\ulT} \mathbb E\left[  \| \Delta_i^n X'\|^2\indicator_{  u_n<\|\tilde \Delta_i^n f\|} \right]
   \leq & \tilde K\Delta_n \sum_{i=1}^{\ulT} \left(\frac{\mathbb E\left[  \|\tilde \Delta_i^n f\|^2 \right]}{u_n^2}\right)^{\frac 12}
   \leq \tilde K^2 \ulT \Delta_n^{1+\frac{1-2w}2},
\end{align*}
which converges to $0$ as $n\to\infty$. Now it remains to prove that the third summand converges to $0$. For $(iii)(b)$ we can use Lemma \ref{L: Bound for truncated jump increments} to obtain a constant $K>0$ and a sequence of real numbers $(\phi_n)_{n\in \mathbb N}$ such that
$$\mathbb E\left[\edit{(iii)^n}(b)\right]\leq K u_n^2 \ulT \Delta_n^{1-2w}\phi_n\leq KT \phi_n,$$
which converges to $0$ as $n\to \infty$ and the claim follows.
\end{proof}

\subsection{Localisation}
In this section we will show how to weaken the Assumptions of Theorems \ref{T: Abstract identification of quadr var},  \ref{T: Identification of continuous and disc quadr var}, \ref{T: rates of convergence}, \ref{T: CLT without disc},  \ref{T: Long-time no disc} by a localizion procedure.   For that, we introduce the localized 
\begin{assumption}\label{As: localised integrals}
There is a constant $A>0$ such that
$$\int_0^{\infty} \|\alpha_s\| ds+\int_0^{\infty} \|\sigma_s\|_{\text{HS}}^2ds+\int_0^{\infty} \int_{H\setminus \{0\}} \|\gamma_s(z)\|^2 F(dz) ds<A.$$
\end{assumption}
\begin{assumption}[r]\label{As: SH}
Assumption \ref{As: H}(r) holds and there is a constant $A>0$ and a deterministic funcion $\Gamma:H\setminus \{0\}\to \mathbb R$ such that  for all $s>0$ and almost all $\omega\in \Omega$
\begin{align*}
\|\alpha_s(\omega)\|+\|\sigma_s(\omega)\|_{\text{HS}}+ \int_{H\setminus\{0\}} \Gamma(z)^rF(dz)<A\quad \text{ and }\quad 
   \|\gamma_s(z)(\omega)\| \leq \Gamma(z).
\end{align*}
\end{assumption}
The next result yields that we can impose the localized conditions for the proofs of the limit theorems. 
\begin{lemma}\label{L: Localisation}
\begin{itemize}
    \item[(i)] Theorem \ref{T: Abstract identification of quadr var}
 holds if it holds under Assumption \ref{As: localised integrals};
 \item[(ii)]  Theorem \ref{T: Identification of continuous and disc quadr var}
 holds if it holds under Assumption \ref{As: SH}(2);
 \item[(iii)]  Theorem \ref{T: rates of convergence}
 holds if it holds under the additionally Assumption \ref{As: SH}(r) for some $r\in (0,1)$;
 \item[(iv)] The following holds under Assumption \ref{As: H}($r$) for  $r<1$
  if it holds under 
 under Assumption \ref{As: SH}($r$):
  $$\sup_{t\in [0,T]}\left\|SARCV_t^n(u_n,-)-\sum_{i=1}^{\ul} (\tilde \Delta_i^n f')^{\otimes 2}\right\|_{\text{HS}}= o_p(\Delta_n^{\frac 12}).$$ 
  \end{itemize}
\end{lemma}

\begin{proof}
Define
\begin{align*}
    \mathcal Z_n^1(t):= & \left\|SARCV_t^n(\infty,-)-\int_0^t \Sigma_s ds+\sum_{s\leq t}(X_s-X_{s-})^{\otimes 2}\right\|_{\text{HS}};\\
     \mathcal Z_n^2(t):= & \left\|SARCV_t^n(u_n,-)-\int_0^t \Sigma_s ds\right\|_{\text{HS}};\\
     \mathcal Z_n^3(t):= & \frac 1{\Delta_n^{\gamma}+b_n^T}\left\|SARCV_t^n(u_n,-)-\int_0^t \Sigma_s ds\right\|_{\text{HS}};\\
     \mathcal Z_n^4(t):= & \frac 1{\Delta_n^{\frac 12}} \left\|SARCV_t^n(u_n,-)-\sum_{i=1}^{\ul} (\tilde \Delta_i^n f')^{\otimes 2}\right\|_{\text{HS}}.
\end{align*}
For any of the above sequences, we can make the observation, that if we have a random sequence $(\mathcal Z_n^i(t))_{t\geq 0}, n\in \mathbb N$ of $\mathbb R$-valued stochastic processes, along with a sequence $(\varphi_N)_{N\in \mathbb N}$ of stopping times such that $\varphi_N \uparrow \infty$ almost surely and such that $(\mathcal Z_n^i(t\wedge \varphi_N))_{n\in \mathbb N}$ 
can be decomposed as
$$\mathcal Z_n^i(t\wedge \varphi_N)\leq A_{n,N}(t)+B_{n,N}(t)$$
where $A_{n,N}^i$ converges to $0$ u.c.p. for all $N\in \mathbb N$ and $B_{n,N}^i$ converges to $0$ as $N\to \infty$ uniformly in $n$ and uniformly on compacts in $t$ we have that
$\edit{\mathcal Z_n^i}\to 0$ u.c.p. as $n\to \infty$. To see that, observe that (with the convention $\sup_{t\in \emptyset}\mathcal Z_n^i(t)=0$)
\begin{align*}
  &  \mathbb P\left[\sup_{t\in [0,T]}\mathcal Z_n^i(t)
    \geq  \epsilon\right]\\
    \leq &\mathbb P\left[\sup_{t\in [0,T]} A_{n,N}^i(t)\geq \epsilon/2\right]+\mathbb P\left[\sup_{t\in [0,T]} B_{n,N}^i(t)\geq \epsilon/2\right]+\mathbb P\left[\sup_{t\in [0,T]}\indicator_{\varphi_N< t}\mathcal Z_n^i(t)\geq \epsilon\right].
\end{align*}
The first summand converges to $0$ as $n\to \infty$ and the second summand converges to $0$ as $N\to\infty$ uniformly in $n$ as does the third summand, since 
\begin{align*}
   \mathbb P\left[\sup_{t\in [0,T]}\indicator_{\varphi_N< t}\mathcal Z_n^i(t)\geq \epsilon\right]\leq \mathbb P\left[\varphi_N< T\right]\to 0 \quad \text{ as } N\to\infty.
\end{align*}
Now we describe how to specify $\tau_N, A_{n,N}^i$ and $B_{n,N}^i$ for $i=1,...,4$.
For any sequence of stopping times $\varphi_N$ we have
\begin{align}\label{Eq in proof of localization addtional semigroup term}
 &   SARCV_{t\wedge \varphi_N}(u_n,-)\notag\\
     = & \sum_{i=1}^{\lfloor t/\Delta_n\rfloor}(\tilde \Delta_i^n (f_{\cdot\wedge \varphi_N})^{\otimes 2}\indicator_{g_n(\tilde\Delta_i^n (f_{\cdot\wedge \varphi_n}))\leq u_n}\\
    &-\left((\lfloor t/\Delta_n\rfloor-\lfloor t\wedge \varphi_N/\Delta_n\rfloor-1)\vee 0\right) \left( f_{\varphi_N}-\mathcal S(\Delta_n)f_{\varphi_N-\Delta_n}\right)^{\otimes 2}\indicator_{g_n(f_{\varphi_N}-\mathcal S(\Delta_n)f_{\varphi_N-\Delta_n})\leq u_n}\notag
\end{align}
and 
\begin{align*}
\int_0^{t\wedge\varphi_N}\Sigma_s ds= & \int_0^t (\indicator_{s\leq \varphi_N}\sigma_s)(\indicator_{s\leq \varphi_N}\sigma_s)^* ds, 
\end{align*}
and
\begin{align*}
\sum_{s\leq t\wedge \varphi_N}(X_{s}-X_{(s)-})^{\otimes 2}= & \sum_{s\leq t}(\indicator_{s\leq \varphi_N}X_{s}-\indicator_{s\leq \varphi_N}X_{s-})^{\otimes 2}.    
\end{align*}
As the initial condition $f_0$ has no impact on the validity of the limit theorems stated in the paper, we can assume w.l.o.g. that $f_0\equiv 0$. Moreover, in this case, it is
 by Theorems \cite[Remark 2.3.10]{liu2015} and \cite[Proposition 3.5]{Knoche05}
\begin{align*}
    f_{t\wedge \varphi_N} 
= &\int_0^{t}\mathcal S(t-s)\alpha_s \indicator_{s\leq \varphi_N}ds+\int_0^{t}\mathcal S(t-s)\sigma_s \indicator_{s\leq \varphi_N}dW_s\\
&+\int_0^{t}\int_{H\setminus \{0\}}\mathcal S(t-s)\gamma_s(z)\indicator_{s\leq \varphi_N}(N-\nu)(dz,ds)
\end{align*}
which is a process of the form \eqref{mild Ito process} with coefficients $\alpha_s':=\alpha_s \indicator_{s\leq \varphi_N}$, $\sigma_s':=\sigma_s \indicator_{s\leq \varphi_N}$ and $\gamma_s(z)':=\gamma_s(z) \indicator_{s\leq \varphi_N}$.

We can hence choose
\begin{align*}
   A_{n,N}^1:= & \left\| \sum_{i=1}^{\lfloor t/\Delta_n\rfloor}(\tilde \Delta_i^n (f_{\cdot\wedge \varphi_N})^{\otimes 2}
   -\int_0^t (\indicator_{s\leq \varphi_N}\sigma_s)(\indicator_{s\leq \varphi_N}\sigma_s)^* ds-\sum_{s\leq t}(X_{s\wedge \varphi_N}-X_{(s\wedge \varphi_N)-})^{\otimes 2}\right\|_{\text{HS}};\\
    A_{n,N}^2:= & \left\| \sum_{i=1}^{\lfloor t/\Delta_n\rfloor}(\tilde \Delta_i^n (f_{\cdot\wedge \varphi_N})^{\otimes 2}\indicator_{g_n(\tilde\Delta_i^n (f_{\cdot\wedge \varphi_n}))\leq u_n}
   -\int_0^t (\indicator_{s\leq \varphi_N}\sigma_s)(\indicator_{s\leq \varphi_N}\sigma_s)^* ds\right\|_{\text{HS}};\\
     A_{n,N}^3:= &\frac 1{\Delta_n^{\gamma}+b_n^T} \left\| \sum_{i=1}^{\lfloor t/\Delta_n\rfloor}(\tilde \Delta_i^n (f_{\cdot\wedge \varphi_N})^{\otimes 2}\indicator_{g_n(\tilde\Delta_i^n (f_{\cdot\wedge \varphi_n}))\leq u_n}
   -\int_0^t (\indicator_{s\leq \varphi_N}\sigma_s)(\indicator_{s\leq \varphi_N}\sigma_s)^* ds\right\|_{\text{HS}};\\
   A_{n,N}^4:= &\frac 1{\Delta_n^{\frac 12}} \left\| \sum_{i=1}^{\lfloor t/\Delta_n\rfloor}(\tilde \Delta_i^n (f_{\cdot\wedge \varphi_N})^{\otimes 2}\indicator_{g_n(\tilde\Delta_i^n (f_{\cdot\wedge \varphi_n}))\leq u_n}
   -\sum_{i=1}^{\lfloor t/\Delta_n\rfloor}(\tilde \Delta_i^n (f'_{\cdot\wedge \varphi_N})^{\otimes 2}\right\|_{\text{HS}}
\end{align*}
and for $i=1,...,4$ define
$$B_{n,N}^i(t):=\frac 1{\Delta_n^{\frac 32}}\left(\left( t- \varphi_N\right)\vee 0\right) \left\|f_{\varphi_N}-\mathcal S(\Delta_n)f_{\varphi_N-\Delta_n})\right\|^2.$$

We first prove that $B_{n,N}^i(t)$ converges to $0$ uniformly in $n$ and uniformly on compacts in $t$. Observe that
$$\lim_{N\to\infty}\sup_{n\in \mathbb N, t\in [0,T]}B_{n,N}^i(t)
=\lim_{N\to\infty}\begin{cases}
    0 & \varphi_N >T\\
    \infty & \varphi_N \leq T,
\end{cases}$$
which converges to $0$ as $N\to \infty$ almost surely, as $\varphi_N\to \infty$ almost surely.

By Assumption, $A_{n,N}^i\to 0$ u.c.p. in $t$ as $n\to \infty$ for $i=1,....,4$, if the coefficients $\alpha_s \indicator_{s\leq \varphi_N^i}$, $\sigma_s \indicator_{s\leq \varphi_N^i}$ and $\gamma_s(z) \indicator_{s\leq \varphi_N^i}$ satisfy the localized conditions in (i),(ii),(iii) and (iv).
Hence, it remains to find localizing sequences of stopping times $\varphi_N^i$ for $i=1,...,4$ such that this is the case. 

We start with the sequence for case (i).
Since we have
$$\mathbb P\left[\int_0^t \left( \|\alpha_s\|+\|\sigma_s\|_{\text{HS}}^2+\int_{H\setminus \{0\}}\|\gamma_s(z)\|^2F(dz)\right)ds<\infty\right]=1$$
we can define
$$\varphi_N:=\inf\left\{ t>0: \int_0^t\left( \|\alpha_s\|+\|\sigma_s\|_{\text{HS}}^2+\int_{H\setminus \{0\}}\|\gamma_s(z)\|^2F(dz)\right)ds>N\right\}.$$
By the right-continuity of $(\mathcal F_t)_{t\geq 0}$, this is a sequence of stopping times,  which is increasing and $\varphi_N\to \infty$ almost surely. Moreover, we obviously have that $\alpha_s \indicator_{s\leq \varphi_n}$, $\sigma_s \indicator_{s\leq \varphi_n}$ and $\gamma_s(z) \indicator_{s\leq \varphi_n}$ satisfy Assumption \ref{As: localised integrals}.

Now we proceed with deriving the localizing sequence for the cases (ii), (iii) and (iv). In this case Assumption \ref{As: H} holds either for $r=2$ or for some $r\in (0,1)$.
By the local boundedness assumption, there is a localizing sequence of stopping times $(\rho_n)_{n\in\mathbb N}$ such that $\alpha_{t\wedge \rho_n}$ is bounded for each $n\in \mathbb N$. As $\sigma$ and $f$ are c{\`a}dl{\`a}g, the sequence of stopping times $\theta_n:=\inf\{s: \|f_s\|+\|\sigma_s\|_{\text{HS}}\geq n\}$ are localizing as well. If $(\tau_n)_{n\in\mathbb N}$ is the sequence of stopping times for the jump part as described in Assumption \ref{As: H}, we can define 
$$\varphi_n:=\rho_n\wedge \theta_n\wedge \tau_n,\quad n\in\mathbb N.$$
This defines a localizing sequence of stopping times, for which  the coefficients $\alpha_s \indicator_{s\leq \varphi_n}$, $\sigma_s \indicator_{s\leq \varphi_n}$ and $\gamma_s(z) \indicator_{s\leq \varphi_n}$ satisfy Assumption \ref{As: SH}(r). This proves the claim.
\end{proof}

\section{Simulating from a 1D first-order SPDE}\label{Sec: description of simulation scheme}
\edit{
We here describe the simulation scheme used in Section~\ref{Sec: Simulation Study} to derive the discrete samples 
$$
f_{\frac in}\left(\frac jn\right), \qquad i,j=1,\dots,n+1,
$$
for $n=100$, from the stochastic PDE 
\begin{align*}
df_t(x) = \partial_x f_t(x)\,dt + dX_t(x).
\end{align*}
The semigroup of the deterministic part consists of left-shifts, 
$$
\mathcal S(t)f(x)=f(x+t), \qquad x,t\geq 0.
$$
The driving semimartingale is assumed to be the sum of a $Q$-Wiener process $W^Q$ and a compound Poisson process $J$.  
For $Q$ we consider two different scenarios, one that generates smooth innovations and one that produces rough innovations in $L^2(0,1)$.  
Precisely, we specify $Qf(x)= \int_0^{\infty} q(x,y) f(y)\,dy$ via its corresponding integral kernel, which in our cases depends only on the difference $x-y$:  
\begin{align*}
q(x,y) &= e^{-|x-y|}, \qquad \forall x,y \in [0,2]\qquad\text{(rough scenario)}, \\
q(x,y) &= e^{-(x-y)^2}, \qquad\forall x,y \in [0,2]\qquad \text{(smooth scenario)}.
\end{align*}
} 

\edit{
The mild solution of $f$ can be decomposed into independent continuous and jump parts:
$$
f_t=f_t'+f_t'' := \int_0^t \mathcal S(t-u)\, dW_u^Q \;+\; J_t. 
$$
Restricted to $[0,1]\times [0,1]$, $f'$ is a Gaussian process.  
Since $q(x,y)=\tilde q(x-y)$ is translation invariant, its spatio-temporal covariance structure reduces to 
$$
\mathbb E[f_t(x)f_s(y)]= \int_0^{\min(t,s)} \tilde q\big((x+t-u)-(y+s-u)\big)\,du 
= \min(t,s)\,\tilde q(x-y+t-s).
$$
By the multiparameter Kolmogorov-Chentsov theorem (\cite[Theorem 3.5]{DPZ2014}), there exists a version of $f'$ which is continuous on $[0,1]\times [0,1]$.  
The jump part $J_t$ consists of countably many spatially constant jumps, and therefore has a version which is right-continuous with left limits in time and continuous in space.  
This justifies the discretization of the solution by pointwise evaluation of this continuous version on a regular grid.  
}

\edit{
Moreover, we can simulate exactly from $f$ via the recursive scheme
$$
f_{\frac{(i+1)}n}\left(\frac jn\right)= \left(\mathcal S\left(\frac 1n\right) f_{\frac in}\right)\left(\frac jn\right)+ N^i_j=f_{\frac in}\left(\frac{(j+1)}n\right)+ \mathbf w^i_j,
$$
where, for each $i=1,\dots,n-1$, the vector $ \mathbf w^i=(N^i_j)_{j=1}^n$ is an $n$-dimensional Gaussian random vector with mean zero and covariance matrix
$$
\mathbf{q}_{j,j'} := n^{-1}\, q\!\left(\tfrac{j-j'}{n}\right).
$$
The increments $(N^i)_{i\geq 1}$ are independent across $i$.  
}

\edit{
To account for the boundary at $x=1$, we need to simulate initially for $j=1,\dots,2n$ when $i=1$, then for $j=1,\dots,n-1$ when $i=2$, for $j=1,\dots,n-2$ when $i=3$, and so on, ensuring that at each time step we obtain a complete sample for $j=1,\dots,n$ and $i=1,\dots,n$.
}

\end{appendix}

\end{document}